\documentclass[12pt]{article}
\usepackage{amsmath}
\usepackage{amssymb}
\usepackage{amsthm}
\usepackage{graphicx}
\usepackage{lineno}
\usepackage{float}
\usepackage[all]{xy}
\usepackage{mathrsfs}

\usepackage{tikz}
\usetikzlibrary{tqft}
\usepackage{xcolor}

\newtheorem{theorem}{Theorem}

\newtheorem{corollary}{Corollary}
\newtheorem{definition}{Definition}

\begin{document}

\title{\bf Control flow in active inference systems}

\author{{Chris Fields$^{a}$\footnote{Corresponding author at: Allen Discovery Center at Tufts University, Medford, MA 02155 USA; {\it E-mail address}: fieldsres@gmail.com}, Filippo Fabrocini$^{b,c}$, Karl Friston$^{d,e}$, James F. Glazebrook$^{f,g}$,}\\ 
{Hananel Hazan$^{a}$, Michael Levin$^{a,h}$ and Antonino Marcian\`{o}$^{i,j,k}$}\\ \\
{\it$^a$ Allen Discovery Center at Tufts University, Medford, MA 02155 USA}\\
{\it$^b$ College of Design and Innovation, Tongji University, 281 Fuxin Rd,}\\
{\it 200092 Shanghai, CHINA}\\
{\it$^c$ Institute for Computing Applications ``Mario Picone'',}\\
{\it Italy National Research Council, Via dei Taurini, 19, 00185 Rome, ITALY}\\
{\it$^d$ Wellcome Centre for Human Neuroimaging, University College London,} \\
{\it London, WC1N 3AR, UK} \\
{\it$^e$ VERSES Research Lab, Los Angeles, CA, 90016 USA} \\
{\it$^f$ Department of Mathematics and Computer Science,} \\
{\it Eastern Illinois University, Charleston, IL 61920 USA} \\
{\it$^g$ Adjunct Faculty, Department of Mathematics,}\\
{\it University of Illinois at Urbana-Champaign, Urbana, IL 61801 USA}\\
{\it$^h$ Wyss Institute for Biologically Inspired Engineering at Harvard University,}\\
{\it Boston, MA 02115, USA}\\
{\it$^i$ Center for Field Theory and Particle Physics \& Department of Physics} \\
{\it Fudan University, Shanghai, CHINA} \\
{\it$^j$ Laboratori Nazionali di Frascati INFN, Frascati (Rome), ITALY}\\
\it$^k$ INFN sezione Roma ``Tor Vergata", I-00133 Rome, ITALY}

\maketitle

{\bf Abstract} \\
Living systems face both environmental complexity and limited access to free-energy resources.  Survival under these conditions requires a control system that can activate, or deploy, available perception and action resources in a context specific way.  We show here that when systems are described as executing active inference driven by the free-energy principle (and hence can be considered Bayesian prediction-error minimizers), their control flow systems can always be represented as tensor networks (TNs).  We show how TNs as control systems can be implmented within the general framework of quantum topological neural networks, and discuss the implications of these results for modeling   biological systems at multiple scales. \\


\tableofcontents

\section{Introduction}

Living things offer remarkable examples of complex, multi-level control policies that guide adaptive function at several scales. At the same time, they are made of components which are usually thought of as physical objects obeying simple rules; how can these two perspectives be unified in a rigorous manner?  The framework of {\em active inference} answers this question be providing a completely general, scale-free formal framework for describing interactions between physical systems in cognitive terms.  It is based on the Free Energy Principle (FEP), first introduced in neuroscience \cite{friston:05, friston:06, friston:07, friston:10, friston:13} before being extended to living systems in general \cite{friston:17, ramstead:18, ramstead:19, kuchling:20} and then to all self-organizing systems \cite{friston:19, ramstead:22, ffgl:22, friston:22}.  The FEP states that any system that interacts with its environment weakly enough to maintain its identifiability over time 1) has a Markov blanket (MB) that separates its internal states from the states of its environment \cite{pearl:88, clark:17, kirchhoff:18, parr:20, sakth:22} and 2) behaves over time in a way that asymptotically minimizes a variational free energy (VFE) measured at its MB.  Equivalently, the FEP states that any system with a non-equilibrium steady-state (NESS) solution to its density dynamics (and hence an MB) will act so as to maintain its state in the vicinity of its NESS.  Any system compliant with the FEP can be described as engaging, at all times, in active inference: a cyclic process in which the system observes its environment, updates its probabilistic ``Bayesian beliefs'' (i.e., posterior or conditional probability densities) over future behaviors, and acts on its environment so as to test its predictions and gain additional information.  The internal dynamics of such a system can be described as inverting a generative model (GM) of its environment that furnishes predictions of the consequences of its actions on its MB.

As a fully-general principle, the FEP applies to all physical systems, not just to behaviorally interesting, plausibly cognitive systems, such as organisms or autonomous robots \cite{friston:19}.  Intuitively, behavior is interesting -- to external observers and, we can assume, to the behaving system itself -- when it is complex, situation-appropriate, and robust in the face of changing environmental conditions.  Friston et al. \cite{friston:22} characterize interesting systems as ``strange particles'', whose internal (i.e., cognitive) states are influenced by their actions only via perceived environmental responses; such systems have to ``ask questions'' of their environments in order to get answers \cite{wheeler:89}.  Such systems, even bacteria and other basal organisms \cite{levin:19, levin:21c, fgl:21, levin:22}, have multiple ways of observing and acting upon their environments and deploy these resources in context-sensitive ways.  In operations-research language, they exhibit situational awareness, i.e., awareness of the context of actions \cite{endsley:12}, and deploy attention systems to manage the informational, thermodynamic, and metabolic costs of maintaining such awareness \cite{ffgl:22, fgl:21}.  Situational awareness is dependent on both short- and long-term memory, or more technically, on the period of time over which precise [Bayesian] beliefs exist, sometimes referred to as the temporal depth or horizon of the GM \cite{levin:19, levin:21c}.  Upper limits can, therefore, be placed on behavioral complexity by examining the capacity and control of memory systems from the cellular scale \cite{fl:17} upwards.  Living systems from microbial mats to human societies employ stigmergic memories \cite{fgl:21} and hence have ``extended minds'' \cite{clark:98} in the sense of the literature on embodied, embedded, enactive, extended, and affective (4EA) cognition \cite{anderson:03, froese:09}.  Such memories must be both readable and writable; hence any system using them must have dedicated, memory-specific perception--action capabilities.

Any system with multiple perception--action (or stimulus--response) capabilities requires a control system that enables context-guided perception and action and precluding the continuous, simultaneous deployment of all available perception--action capabilities.  Such self organization entails the selection of a particular course of action -- i.e., policy -- from all plausible policies entertained by the system's GM.  In the active inference framework, the system's internal states -- hence its GM -- can be read as encoding posterior probability densities (i.e., Bayesian beliefs) over the causes of its sensory states, including, crucially, its own actions. This leads to the notion of planning and control as inference \cite{attias:03, botvinick:12, lanillos:21}, with the ensuing selection of an action given by the most likely policy.  In bacteria such as {\em E. coli}, for example, mutual inhibition between gene regulatory networks (GRNs) for different metabolic operons permit the expression of specific carbon-source (e.g., sugar) metabolism pathways only when the target carbon source is detected in the environment \cite{chubukov:14}.  The control of foraging behavior via chemotaxis employs a similar, in this case bistable, mechanism \cite{micali:16}.  Such mechanisms are active in multicellular morphogenesis, for example, in the head-versus-tail morphology decision in planaria \cite{pezzulo:21}. In the human brain, mutual inhibition between competing visual processing streams is evident in binocular rivalry (switching between distinct scenes presented to left and right eyes) or in the changing interpretations of ambiguous figures such as the Necker cube \cite{blake:02, stertzer:09}; similar competitive effects are observed in other sensory pathways \cite{schwartz:12}.  It also characterizes the competitive interaction between the dorsal and ventral attention systems, which implement top-down and bottom-up targeting of sensory resources, respectively \cite{vossel:14}.  It is invoked at a still larger scale in global workspace models of conscious processing, in which incoming information streams must compete, with each inhibiting the others, for ``access to consciousness'' \cite{baars:03, baars:13}.  Mutual inhibition creates an energetic barrier that the control system that implements switching must expend free-energy resources to overcome; the controller must not only turn ``on'' the preferred system, but also turn ``off'' the inhibition.  The required free energy expenditure in turn induces hysteresis and hence the non-linear, winner-takes-all ``switch'' behavior in the time regime.  Such barriers and their temporal consequences persist in more complex control systems whenever two perception--action capabilities are either functionally incompatible or too expensive to deploy simultaneously.

Switching between perception--action capabilities can be regarded, from a theoretical, FEP perspective, as selecting a plausible policy, or plan, supported by the GM.  Technically, the probability distribution over policies or plans can be computed from a free energy functional expected under the posterior predictive density over possible outcomes, as described in \S \ref{attractor} below.  The control system that implements the switching process can be considered to employ the GM to predict, or assign probability distribution to, each perception-action capability (i.e., policy) as a function of context \cite{kfl:22, parr:19}.  We can consider the GM to generate probabilistic ``beliefs'' about the consequences of actions, where here a ``belief'' is just a mathematically-described structure, e.g., a classical conditional probability density or a quantum state with an assigned amplitude.  ``Planning'' or ``control'' can, therefore, always be cast as inference -- again in the basal sense of computation -- implemented by variational message passing or ``belief propagation'' on a (normal style) factor graph: a graph with nodes corresponding to the factors of a probability distribution and undirected edges corresponding to message-passing channels.  Factor graphs can be combined with message passing schemes, with the messages generally corresponding to sufficient statistics of the factors in question, to provide an efficient computation of functions such as marginal densities \cite{winn:05, dauwels:07}.  Hence one can formalize control -- under the FEP -- in terms of control as inference, which implies that there is a description of control in terms of message passing on a factor graph.  When the GM is over discrete states, this implies a description of control in terms of tensor operators. 

Nearly all simulations of planning -- under discrete state space GMs -- use the factor-graph formalism. Crucially, the structure of the factor graph embodies the structure of the GM and, effectively, the way that any system represents the (apparent causes of) data on its MB; i.e., the way it ``carves nature at its joints,'' into states, objects and categorical features.  Under the (classical) FEP, the factors that constitute the nodes of the factor graph correspond to the state-space factorization in a mean field approximation, as used by physicists, or by statisticians to implement variational Bayesian (a.k.a., approximate Bayesian) inference \cite{parr:20a}.  See \cite{dacosta:20} for technical details, \cite{friston:17a} for an application to the brain, and Supplementary Information, Table 1 for a list of selected applications.

We show in this paper that control flow in such systems can always be formally described as a tensor network, a factorization of some overall tensor (i.e., high-dimensional matrix) operator into multiple component tensor operators that are pairwise contracted on shared degrees of freedom  \cite{orus:19}.  In particular, we show that the factorization conditions that allow the construction of a TN are exactly the same as those that allow the identification of distinct, mutually conditionally independent (in quantum terms, decoherent), sets of data on the MB, and hence allow the identification of distinct ``objects'' or ``features'' in the environment.  This equivalence allows the topological structures of TNs -- many of which have been well-characterized in applications of the TN formalism to other domains \cite{orus:19} -- to be employed as a classification of control structures in active inference systems; including cells, organisms, and multi-organism communities.  It allows, in particular, a principled approach to the question of whether, and to what extent, a cognitive system can impose a decompositional or mereological (i.e., part-whole) structure on its environment.  Such structures naturally invoke a notion of locality, and hence of geometry.  The geometry of spacetime itself has been described as a particular TN -- a multiscale entanglement renormalization ansatz (MERA) \cite{bao:17, hu:17, chandra:21} -- suggesting a deep link between control flow in systems capable of observing spacetime (i.e., capable of implementing internal representations of spacetime) and the deep structure of spacetime as a physical construct.

We begin in \S \ref{formalism} by analyzing the control-flow problem in three different representations of active inference.  First, we employ the classical, statistical formulation of the FEP \cite{friston:19, ramstead:22} in \S \ref{attractor} to describe control flow as implementing discrete, probabilistic transitions between dynamical attractors on a manifold of computational states.  We then reformulate the physical interaction in quantum information-theoretic terms in \S \ref{qrf}; in this formulation \cite{ffgl:22}, components of the GM can be considered to be distinct quantum reference frames (QRFs) \cite{aharonov:84, bartlett:07} and represented by hierarchical networks of Barwise-Seligman classifiers \cite{barwise:97} as developed in \cite{fg:19a, fg:20a, fg:21, fgm:21}.  Control flow then implements discrete transitions between QRFs.  The third step, in \S \ref{tqft}, employs the mapping between hierarchies of classifiers and topological quantum field theories (TQFTs) developed in \cite{fgm:22a}.  Here, control flow is implemented by a TQFT, with transition amplitudes given by a path integral.  The second and third of these representations provide formal characterizations of intrinsic (or ``quantum'') context effects that are consistent with both the sheaf-theoretic treatment of contextuality in \cite{abramsky:11, abramsky:17} and the Contextuality by Default (CbD) approach of \cite{Dzha2017a, Dzha2018}; see also the discussion in \cite{fg:21} and \cite[\S7.2]{fgm:22a}.  The underlying theme is that contextuality arises due to the non-existence of any globally definable (maximally connected) conditional probability distribution across all possible observations (see e.g., \cite{adlam:21} for a review from a more general physics perspective).  Extending our earlier analysis \cite{fg:21}, we discuss reasons to expect that active inference systems will generically exhibit such context effects. 

We then develop in \S \ref{tn} a fully-general tensor representation of control flow, and prove that this tensor can be factored into a TN if, and only if, the separability (or conditional statistical independence) conditions needed to identify distinct features of or objects in the environment are met.  We show how TN architecture allows classification of control flows, and give two illustrative examples.   We discuss in \S \ref{tqnn} several established relationships between TNs and artificial neural network (ANN) architectures, and how these generalize to topological quantum neural networks \cite{fgm:22a, tqnn:22a}, of which standard deep-learning (DL) architectures are a classical limit \cite{tqnn:22b}.  We turn in \S \ref{applications} to implications of these results for biology, and discuss how TN architectures correlate with the observational capabilities of the system being modeled, particularly as regards abilities to detect spatial locality and mereology.  We consider how to classify known control pathways in terms of TN architecture and how to employ the TN representation of control flow in experimental design.  We conclude by looking forward to how these FEP-based tools can further integrate the physical and life sciences.

\section{Formal description of the control problem} \label{formalism}

\subsection{The attractor picture} \label{attractor}

Let $U$ be a random dynamical system that can be decomposed into subsystems with states $\mu(t)$, $b(t)$, and $\eta(t)$ such that the dependence of the $\mu(t)$ on the $\eta(t)$, and vice-versa, is only via the $b(t)$.  In this case, the $b(t)$ form an MB separating the $\mu(t)$ from the $\eta(t)$.  We will refer to the $\mu(t)$ as ``internal'' states, to the $\eta(t)$ as ``environment'' states, and to the combined $\pi(t) = (b(t), \mu(t))$ as ``particular'' (or ``particle'') states \cite{friston:19}.  The FEP is a variational or least-action principle stating that any system -- that interacts sufficiently weakly with its environment -- can be considered to be enclosed by an MB, i.e. any ``particle'' with states $\pi(t) = (b(t), \mu(t))$, will evolve in a way that tends to minimize a variational free energy (VFE) $F(\pi)$ that is an upper bound on (Bayesian) surprisal.  This free energy is effectively the divergence between the variational density encoded by internal states and the density over external states conditioned on the MB states.  It can be written \cite[Eq. 2.3]{friston:19},

\begin{equation} \label{VFE-def}
\begin{aligned}
F(\pi ) & = \underbrace{{{\mathbb{E}}_{q(\eta )}}[\ln {{q}_{\mu }}(\eta )-\ln p(\eta, b)]}_{\text{Variational free energy}} \\
 & =\underbrace{{{\mathbb{E}}_{q}}[-\ln p(b|\eta )-\ln p(\eta )]}_{\text{Energy constraint (likelihood  }\!\!\And\!\!\text{  prior)}}-\underbrace{{{\mathbb{E}}_{q}}[-\ln {{q}_{\mu }}(\eta )]}_{\text{Entropy}} \\
 & =\underbrace{{{D}_{KL}}[{{q}_{\mu }}(\eta )|p(\eta )]}_{\text{Complexity}}-\underbrace{{{\mathbb{E}}_{q}}[\ln p(b|\eta )]}_{\text{Accuracy}} \\
 & =\underbrace{{{D}_{KL}}[{{q}_{\mu }}(\eta )||p(\eta |b)]}_{\text{Divergence}}\underbrace{-\ln p(b)}_{\text{ Log evidence}}\ge -\ln p(b)
\end{aligned}
\end{equation}
\noindent
The VFE functional $F(\pi)$ is an upper bound on surprisal (a.k.a. self-information) $\mathfrak{I}(\pi) = -\ln p(\pi) > -\ln p(b)$ because the Kullback-Leibler divergence term ($D_{KL}$) is always non-negative. This KL divergence is between the density over external states $\eta$, given the MB state $b$, and a variational density $Q_{\mu} (\eta )$ over external states parameterized by the internal state $\mu$.  If we view the internal state $\mu$ as encoding a posterior over the external state $\eta$, minimizing VFE is, effectively, minimizing a prediction error, under a GM encoded by the NESS density. In this treatment, the NESS density becomes a probabilistic specification of the relationship between external or environmental states and particular (i.e., ``self'') states.  We can interpret the internal and active MB states in terms of active inference, i.e., a Bayesian mechanics \cite{ramstead:22}, in which their expected flow can be read as perception and action, respectively.  Here ``active'' states are a subset of the MB states that are not influenced by environmental states and -- for the kinds of particles considered here -- do not influence internal states.  In other words, active inference is a process of Bayesian belief updating that incorporates active exploration of the environment.  It is one way of interpreting a generalized synchrony between two random dynamical systems that are coupled via an MB.

If the ``particle'' $\pi$ is a biological cell, it is natural to consider the MB $b$ to be implemented by the cell membrane and the ``internal'' states $\mu$ to be the internal macromolecular or biochemical states of the cell; indeed, it is this association that motivated the application of the FEP to cellular life \cite{friston:13}.  In this case, the NESS corresponds to the state, or neighborhood of states, that maintain homeostasis (or more broadly, allostasis \cite{sterling:88, barrett:16, corcoran:20}) and hence maintain the structural and functional integrity of $\pi$ as a living cell.  This activity of self-maintenance has been termed ``self-evidencing'' \cite{hohwy:16}; systems compliant with the FEP can be considered to be continually generating evidence of -- or for -- their continued existence \cite{friston:19}.  

In the terminology of \cite{friston:22} cells are ``strange particles'' -- their signal transduction pathways monitor (components of) the states of their environments, but do not directly monitor their actions on their environments (i.e., their own active states).  The consequences of any action can only, therefore, be deduced from the response of the environment.  In this situation, causation is always uncertain: whether an action by the environment on the cell -- what the cell detects as an environmental state change -- is a causal consequence of an action the cell has taken in the past cannot be determined by the data available to the cell.  Every action, therefore, increases VFE, while every observation (potentially) decreases it.  The (apparent) task of the cell's GM is to minimize the increases, on average, while maximizing the decreases.  

The Bayesian mechanics afforded by the FEP implies a (classical) thermodynamics; indeed, the FEP can be read as a constrained maximum entropy or caliber principle \cite{sakth:22a, sakth:22b} (Sakthivadivel 2022, Sakthivadivel 2022).  This follows from the fact that inference, i.e., self evidencing,  entails belief updating and belief updating incurs a thermodynamic cost via the Jarzynski equality \cite{landauer:61, jarzynski:97, evans:03}.  This  cost provides a lower bound on the  thermodynamic free energy required for metabolic maintenance.  For example, a cell's actions on its environment -- e.g., chemotactic locomotion -- are largely driven by the need to acquire thermodynamic free energy.  The cell's GM cannot, therefore, minimize VFE by minimizing action \cite{dark-room:12}; instead, it must successfully predict which actions will replenish its free-energy supply.  As actions are energetically expensive, this requires trading off short-term costs against long-term goals.  As shown in \cite{kfl:22}, selective pressures operating on different timescales favor the development of metaprocessors that control lower-level actions in a context-dependent way; these are often implemented via a hierarchical GM \cite{pezzulo:15}.  Such meta-level control provides probabilistic models of risk-sensitive actions in context.  

While such systems may be described as regulating free-energy seeking actions, they also regulate information-seeking actions, i.e., curiosity-driven exploration \cite{friston:15a, schmidhuber:91, sun: 11}.  This follows because VFE provides an upper bound on complexity minus accuracy \cite{sengupta:18}. The expected free energy (EFE), conditioned upon any action, can therefore be scored in terms of expected complexity and expected inaccuracy. Expected complexity is ``risk'' and corresponds to the degree of belief updating that incurs a thermodynamic cost; leading to risk-sensitive control (e.g., phototropism). Expected inaccuracy corresponds to ``ambiguity'' leading to epistemic behaviors (e.g., searching for lost keys under a streetlamp) \cite{parr:19}.

When context-dependent control is considered, the neighborhood of the NESS resolves into a network of local minima corresponding to fixed perception-action loops separated by energetic barriers that the control system must overcome to switch between loops.  For example, in a cell, this energetic barrier comprises the energy required to activate one pathway while de-activating another, which may include the energetic costs of phosphorylation, other chemical modifications, additional gene expression, etc.  Different pairs of pathways can be expected to be separated by energetic barriers of different heights, generating a topographically-complex free energy landscape that coarse-grains, in a long-time average, to the neighborhood of the NESS, i.e., to the maintenance of allostasis \cite{barrett:16, corcoran:20, seth:16}.

As noted earlier, we can think of controllable perception-action loops as nodes on a factor graph, with the edges corresponding to pathways for control flow, and the transition probabilities labeling the edges as inversely proportional to the energetic barrier between loops.  This allows representing the GM for meta-level control (i.e., hierarchical) as a message-passing system as described in \cite{friston:17a}.  The presence of very high energetic barriers can render such a GM effectively one-way, as seen in the context-dependent switches between signal transduction pathways and GRNs that characterize cellular differentiation during morphogenesis.  Biological examples of these include modifications of bioelectric pattern memories in planaria, which can create alternative-species head shapes that eventually remodel back to normal \cite{emmons-bell:15}, or produce 2-headed worms which are permanent, and regenerate as 2-headed in perpetuity \cite{oviedo:10}.

\subsection{The QRF picture} \label{qrf}

Cellular information processing has traditionally been treated as completely classical, i.e., as implemented by causal networks of macromolecules, each of which undergoes classical state transitions via local dynamical processes that are conditionally independent of the states of other parts of the network.  While the ``quantum'' nature of proteins and other macromolecules is broadly acknowledged, the scale at which quantum effects are important remains controversial, with straightforward single-molecule decoherence models predicting decoherence times of attoseconds ($10^{-18}$ s) or less \cite{tegmark:00, schloss:07}: several orders of magnitude below the timescales of processes involved in molecular information processing \cite{zwier:10}.  While functional roles for quantum coherence in intramolecular information processing have been demonstrated, intermolecular coherence remains experimentally elusive \cite{marais:18, cao:20, kim:21, baiardi:22}.  

The free-energy budgets of both prokaryotic and eukaryotic cells are, however, orders of magnitude smaller than would be required to support fully-classical information processing at the molecular scale, suggesting that cells employ quantum coherence as a computational resource \cite{fl:21}.  Indirect evidence of longer-range, tissue-scale coherence in brains has also been reported \cite{kerskens:22}.   Reformulating the FEP in quantum information-theoretic terms enables it to describe situations in which long-range coherence, and hence quantum computation, cannot be neglected. 

Following the development in \cite{ffgl:22}, we consider a bipartite decomposition $U = AB$ of a finite, isolated system $U$ for which the interaction Hamiltonian $H_{AB} = H_U - (H_A + H_B)$ is sufficiently weak over the time period of interest that the joint state $U$ is separable (i.e., factors) as $|U \rangle = |A \rangle |B \rangle$.  In this case, we can choose orthogonal basis vectors $|i^k \rangle$ so that:

\begin{equation} \label{ham}
H_{AB} = \beta_k K_B\, T_k \sum_i^N \alpha^k_i M^k_i,
\end{equation}
\noindent
where $K_B$ denotes Boltzmann's constant, $T$ is the absolute temperature of the environment, $k =~A$ or $B$, the $M^k_i$ are $N$ mutually-orthogonal Hermitian operators with eigenvalues in $\{ -1,1 \}$, the $\alpha^k_i \in [0,1]$ are such that $\sum^N_i \alpha^k_i = 1$, and $\beta_k \geq \ln 2$ is an inverse measure of $k$'s thermodynamic efficiency that depends on the internal dynamics $H_k$; see \cite{fg:20a, fgm:21, fm:20, addazi:21} for further motivation and details of this construction and \cite{fgm:22b} for a pedagogical review.  This description is purely topological, attributing no geometry to either $U$ or $\mathscr{B}$; hence it allows the ``embedding space'' of perceived ``objects'' to be an observer-dependent construct.  It has several relevant consequences:

\begin{itemize}
\item We can regard $A$ and $B$ as separated, and determined by independent measures. They are separated by -- and interact via -- a holographic screen $\mathscr{B}$ that can be represented, without loss of generality, by an array of $N$ non-interacting qubits, where $N$ is the dimension of $H_{AB}$ \cite{fm:20, addazi:21}.
\item $A$ and $B$ can be regarded as exchanging finite $N$-bit strings, each of which encodes one eigenvalue of $H_{AB}$ \cite{fm:20}.
\item $A$ and $B$ have free choice of basis for $H_{AB}$, corresponding to free choice of local frames at $\mathscr{B}$, e.g., free choice, for each qubit $q_i$ on $\mathscr{B}$, of the local $z$ axis and hence the $z$-spin operator $s_z$ that acts on $q_i$ \cite{fgm:22b}.
\item Choice of basis corresponds to choosing the zero-point of total energy) by each of $A$ and $B$.  The systems $A$ and $B$ are, therefore, in general at informational, but not at thermal equilibrium \cite{ffgl:22}. 
\item As $A$ and $B$ must obtain from $B$ or $A$, respectively, whatever thermodynamic free energy is required, by Landauer's principle \cite{landauer:61, landauer:99, bennett:82}, to fund the encoding of classical bits on $\mathscr{B}$ (as well as any other irreversible classical computation), $A$ and $B$ must each devote some sector $F$ of $\mathscr{B}$ to free-energy acquisition.  The bits in $F$ are ``burned as fuel'' and so do not contribute input data to computations.  Waste-heat dissipation by one system is free energy acquisition by the other.  The free-energy sectors $F_A$ and $F_B$ of $A$ and $B$ need not align as subsets of qubits on $\mathscr{B}$; that is, qubits that $A$ regards as free-energy sources may be regarded by $B$ as informative outputs and vice-versa \cite{fg:20a, fgm:21}.
\item The actions of the internal dynamics $H_A$ and $H_B$ on $\mathscr{B}$ can be represented by $A$- and $B$-specific sets of QRFs, each of which both ``measures'' and ``prepares'' qubits on $\mathscr{B}$.  Each QRF acts on the qubits in some specific sector of $\mathscr{B}$, breaking the permutation symmetry of Eq. \eqref{ham} \cite{fg:20a, fgm:21, fgm:22a}.  Only QRFs acting on sectors other than $F$ implement informative computations; we will therefore restrict attention to these QRFs.
\item Each ``computational'' QRF can, without loss of generality, be represented by a cone-cocone diagram (CCCD) comprising Barwise-Seligman classifiers and infomorphisms between them \cite{barwise:97, fg:19a}.  The apex of each such CCCD is, by definition, both the category-theoretic limit and colimit of the ``input/output'' classifiers that correspond, formally, to the operators $M^k_i$ in Eq. \eqref{ham} \cite{fg:20a, fgm:21, fgm:22a}.
\end{itemize}

Typically, a CCCD is structured as a distributed information flow in the form:

\begin{equation}\label{cccd-2}
\begin{gathered}
\xymatrix@C=6pc{\mathcal{A}_1 \ar[r]_{g_{12}}^{g_{21}} & \ar[l] \mathcal{A}_2 \ar[r]_{g_{23}}^{g_{32}} & \ar[l] \ldots ~\mathcal{A}_k \\
&\mathbf{C^\prime} \ar[ul]^{h_1} \ar[u]^{h_2} \ar[ur]_{h_k}& \\
\mathcal{A}_1 \ar[ur]^{f_1} \ar[r]_{g_{12}}^{g_{21}} & \ar[l] \mathcal{A}_2 \ar[u]_{f_2} \ar[r]_{g_{23}}^{g_{32}} & \ar[l] \ldots ~\mathcal{A}_k \ar[ul]_{f_k}
}
\end{gathered}
\end{equation}
\noindent
incorporating sets of classifiers $\{A_{\alpha}\}$ and (logic) infomorphisms $\{f_i, g_{jk}\}$ \cite[Ch 12]{barwise:97} over suitable index ranges.  As a memory-write system, Diagram \eqref{cccd-2} depicts a generic a blueprint for a bow-tie or variational autoencoder (VAE) network amenable to describing a hierarchical Bayesian network with belief-updating as discussed in e.g. \cite{ffgl:22,fg:21,fgm:22a}. Crucially, it is the non-commutativity of CCCDs of this form that specifies intrinsic or quantum contextuality, as occurs, for instance, when the colimit core $\mathbf{C^\prime}$ is undefinable \cite[\S7, \S 8]{fg:21} \cite[\S 7.2]{fgm:22a}. Consequences of such contextualityare discussed via examples in \S\ref{applications}.

The holographic screen $\mathscr{B}$ functions as an MB separating $A$ from $B$.  It can be regarded as having an $N$-dimensional, $N$-qubit Hilbert space $\mathcal{H}_{q_i} = \prod_i q_i$.  While $\mathcal{H}_{q_i}$ is strictly ancillary to $\mathcal{H}_U = \mathcal{H}_A \otimes \mathcal{H}_B$, the classical situation can be recovered in the limit in which the entanglement entropies $\mathcal{S}(|A \rangle), \mathcal{S}(|B \rangle) \rightarrow 0$ by considering the products $\mathcal{H}_A \otimes \mathcal{H}_{q_i}$ and $\mathcal{H}_B \otimes \mathcal{H}_{q_i}$ to be ``particle'' state spaces for $A$ and $B$, respectively.  In this classical limit, the states of $\mathcal{H}_{q_i}$ become the blanket states of an MB that functions as a classical information channel \cite{fm:20, addazi:21, fgm:22b}.  In quantum holographic coding, for example, $\mathscr{B}$ is often represented by a polygonal tessellation of the hyperbolic disc, with qubits represented by polygonal centroids.  A specific TN model of a pentagon code is developed in \cite{pastawski:15}; see in particular their Fig. 4.  The geometric description of $\mathscr{B}$ as implementing holographic coding, and its classical limit as an MB structured as a direct acyclic graph (DAG), is further explored in the setting of TQNNs in \cite{fgm:23}.

In this quantum-theoretic picture, ``systems'' or ``objects'' observed and manipulated by $A$ or $B$ correspond to sectors on $\mathscr{B}$ that are the domains of particular QRFs deployed by $A$ or $B$, respectively \cite{fgm:21, ffgl:22, fgm:22a}.  To simplify notation, we use the same symbol, e.g., `$Q$' to denote both a QRF $Q$ and the sector $dom(Q)$ on $\mathscr{B}$.  Any identifiable system $X$ factors into a ``reference'' component $R$ that maintains a time-invariant state $|R \rangle$ or more generally, state density $\rho_R$, that allows re-identification and hence sequential measurements over extended time, and a ``pointer'' component $P$ with a time-varying state $|P \rangle$ or density $\rho_P$.  It is this pointer component, named for the pointer of an analog instrument, which is the ``state of interest'' for measurements.  The QRFs $R$ and $P$ clearly must commute, and the sectors $R$ and $P$ clearly must be mutually decoherent \cite{fgm:21, ffgl:22, fgm:22a}.  All ``system'' sectors must be components of some overall sector $E$ that corresponds to the ``observable environment.''  The recording of measurement outcomes to a classical memory and the reading of previously-recorded outcomes from memory can similarly be represented by a QRF $Y$.  As $dom(Y)$ is a sector on $\mathscr{B}$, recorded memories of $A$ are exposed to and hence subject to modification by $B$ and vice-versa.  Both the observable environment $E$ and the memory sector $Y$ must be disjoint from, and decoherent with, the free-energy sector $F$.  

As actions on $\mathscr{B}$ encode classical data, they have an associated free energy cost of at least ln2 $K_B T$ per bit \cite{landauer:61, landauer:99, bennett:82} that must originate from the source at $F$.  Time-energy complementary associates a minimum time of $h / [\mathrm{ln} 2 (K_B T)]$, with $h$ being the Planck's constant, to this energy expenditure.  We can, therefore, associate actions on $\mathscr{B}$, including memory writes, with ``ticks'' of an internal time QRF, which we denote $t_A$ and $t_B$ for $A$ and $B$, respectively.  Assuming all observational outcomes are written to memory, we can represent the situation as in Fig. \ref{mem-write-fig}.  The time QRF is effectively an outgoing bit counter that can be represented by a groupoid operator $\mathcal{G}_{ij}: t_i \rightarrow t_j$ \cite{fg:20a}.  As outgoing bits are oriented in opposite directions with respect to $\mathscr{B}$ for $A$ and $B$, the time ``arrows'' $t_A$ and $t_B$ point in opposite directions.  Hence $A$ and $B$ can both be regarded as ``interacting with their own futures'' as discussed in \cite{fgm:22b}.

\begin{figure}[H]
\centering
\includegraphics[width=10 cm]{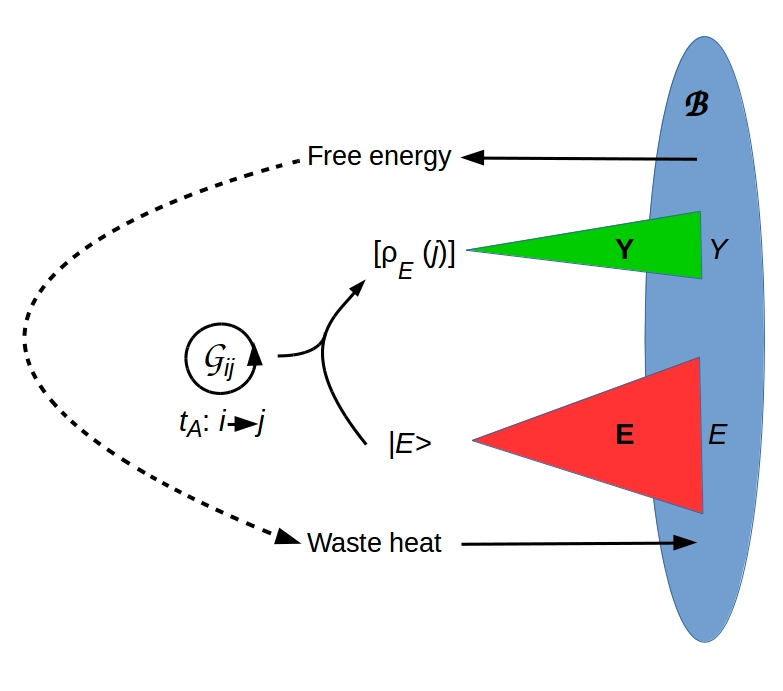}
\caption{Cartoon illustration of QRFs required to observe and write a readable memory of an environmental state $\vert E \rangle$.  The QRFs $\mathbf{E}$ and $\mathbf{Y}$ read the state from $E$ and write it to the memory $Y$ respectively.  Any identified system $S$ must be part of $E$.  The clock $\mathcal{G}_{ij}$ is a time QRF that defines the time coordinate $t_A$.  The dashed arrow indicates the observer's thermodynamic process that converts free energy obtained from the unobserved sector $F$ of $\mathscr{B}$ to waste heat exhausted through $F$.  Adapted from \cite{fgm:21}, CC-BY license.}
\label{mem-write-fig}
\end{figure}

Measurements of a system $X$ can be considered sequential if: 1) they separated in time according to the internal time QRF, and 2) their outcomes are recorded to memory to enable comparability across time.  We show in \cite{fgm:22a} that sequential measurements can always be represented by one of two schemata.  Using the compact notation:

\begin{equation} \label{QRF-1}
\begin{gathered}
\begin{tikzpicture}
\draw [thick] (0,0) -- (2,1) -- (2,-1) -- (0,0);
\node at (1.3,0) {$S$};
\end{tikzpicture}
\end{gathered}
\end{equation}
\noindent
to represent a QRF $S$, we can represent measurements of a physical situation in which one system divides into two, possibly entangled, systems with a diagram of the form:

\begin{equation} \label{QRF-2}
\begin{gathered}
\begin{tikzpicture}
\node at (0,0) {$S$};
\draw [thick] (0.2,0) -- (1,1) -- (2,1.5) -- (2,0.5) -- (1,1);
\node at (1.7,1) {$S_1$};
\draw [thick] (0.2,0) -- (1,-1) -- (2,-0.5) -- (2,-1.5) -- (1,-1);
\node at (1.7,-1) {$S_2$};
\end{tikzpicture}
\end{gathered}
\end{equation}
\noindent
Parametric down-conversion of a photon exemplifies this kind of process.  The reverse process can be added to yield:

\begin{equation} \label{flow-1}
\begin{gathered}
\begin{tikzpicture}
\node at (0,0) {$S$};
\draw [thick] (0.2,0) -- (1,1) -- (2,1.5) -- (2,0.5) -- (1,1);
\node at (1.7,1) {$S_1$};
\draw [thick] (0.2,0) -- (1,-1) -- (2,-0.5) -- (2,-1.5) -- (1,-1);
\node at (1.7,-1) {$S_2$};
\draw [thick] (-2.7,0) -- (-1.7,0.5) -- (-1.7,-0.5) -- (-2.7,0);
\node at (-2,0) {$S$};
\draw [thick, ->] (-1.5,0) -- (-0.3,0);
\draw [thick, ->] (2.2,0) -- (3.4,0);
\draw [thick] (3.6,0) -- (4.6,0.5) -- (4.6,-0.5) -- (3.6,0);
\node at (4.3,0) {$S$};
\end{tikzpicture}
\end{gathered}
\end{equation}
\noindent
In the second type of sequential measurement process, the pointer-state QRF $P$ is replaced with an alternative QRF $Q$ with which it does not commute.  Sequences in which position and momentum, $s_z$ and $s_x$ are measured alternately are examples.  These can be represented by the diagram

\begin{equation} \label{flow-2}
\begin{gathered}
\begin{tikzpicture}
\node at (0,0) {$\mathbf{S}$};
\draw [thick] (0.2,0) -- (1,1) -- (2,1.5) -- (2,0.5) -- (1,1);
\node at (1.7,1) {$\mathbf{P}$};
\draw [thick] (0.2,0) -- (1,-1) -- (2,-0.5) -- (2,-1.5) -- (1,-1);
\node at (1.7,-1) {$\mathbf{R}$};
\draw [thick] (-2.7,0) -- (-1.7,0.5) -- (-1.7,-0.5) -- (-2.7,0);
\node at (-2,0) {$\mathbf{S}$};
\draw [thick, ->] (-1.5,0) -- (-0.3,0);
\draw [thick, ->] (2.2,0) -- (3.4,0);
\node at (3.7,0) {$\mathbf{S}$};
\draw [thick] (3.9,0) -- (4.7,1) -- (5.7,1.5) -- (5.7,0.5) -- (4.7,1);
\node at (5.4,1) {$\mathbf{Q}$};
\draw [thick] (3.9,0) -- (4.7,-1) -- (5.7,-0.5) -- (5.7,-1.5) -- (4.7,-1);
\node at (5.4,-1) {$\mathbf{R}$};
\draw [thick, ->] (5.9,0) -- (7.1,0);
\draw [thick] (7.3,0) -- (8.3,0.5) -- (8.3,-0.5) -- (7.3,0);
\node at (8,0) {$\mathbf{S}$};
\end{tikzpicture}
\end{gathered}
\end{equation}
\noindent
As both $P$ and $Q$ must commute with $R$, the commutativity requirements for $S$ are satisfied.

The sequences of operations depicted in Diagrams \eqref{flow-1} and \eqref{flow-2} clearly raise the questions of how control is implemented, and of how the context changes that drive control flow are detected.  Before turning to these questions in \S \ref{tn}, we review a path-integral representation of QRFs, show that the same representation captures the behavior of any system $X$ identified by a QRF, and discuss the questions of multiple observers and quantum contextuality.

\subsection{The TQFT picture} \label{tqft}

As a least-action principle, the FEP is fundamentally a statement about the paths followed by the joint system $U$ through its state space.  The classical FEP is amenable to a path-integral formulation \cite{friston:22} that expresses the expected value of any observable (functional) $\Omega[x(t)]$ of paths $x(t)$ through the relevant state space as (\cite{seifert:12}, Eq. 6):

\begin{equation} \label{class-path-int}
\langle \Omega[x(t)] \rangle = \int d x_0 \int d [x(t)] \Omega[x(t)] p( x(t) | x_0) p_0 (x_0)
\end{equation}
\noindent
where $x_0$ is the initial state and $p( x(t) | x_0)$ is the conditional probability of the path $x(t)$.  Quantum theory generalizes this expression by, effectively. replacing $\Omega[x(t)]$ with an automorphism on the relevant Hilbert space and $p( x(t) | x_0)$ with an amplitude for $x(t)$ given the initial state $x_0$.  For some finite-dimensional Hilbert space $\mathcal{H}$, the manifold of all such automorphisms is a cobordism on $\mathcal{H}$, which is by definition a TQFT on $\mathcal{H}$ \cite{atiyah:88}. 

We show in \cite{fgm:22a} that any sequential measurement of any sector $X$ of $\mathscr{B}$ induces a TQFT on $X$, considered as a projection of the $N$-dimensional boundary Hilbert space $\mathcal{H}_{q_i}$ associated with $\mathscr{B}$.  In particular, measurement sequences of the form of Diagram \eqref{flow-1} can be mapped to cobordisms, i.e., to manifolds of maps between two designated boundaries, of the form:

\begin{equation} \label{CCCD-to-Cob}
\begin{gathered}
\begin{tikzpicture}[every tqft/.append style={transform shape}]
\draw[rotate=90] (0,0) ellipse (2.5cm and 1 cm);
\node[above] at (0,1.7) {$\mathscr{B}$};
\node at (-0.4,0) {$S$};
\begin{scope}[tqft/every boundary component/.style={draw,fill=green,fill opacity=1}]
\begin{scope}[tqft/cobordism/.style={draw}]
\begin{scope}[rotate=90]
\pic[tqft/cylinder, name=a];
\pic[tqft/pair of pants, anchor=incoming boundary 1, name=b, at=(a-outgoing boundary 1)];
\end{scope}
\end{scope}
\end{scope}
\draw[rotate=90] (0,-4) ellipse (2.5cm and 1 cm);
\node[above] at (4,1.7) {$\mathscr{B}$};
\node at (2,0.7) {$\mathscr{S}$};
\node at (4.4,1.1) {$S_1$};
\node at (4.4,-1.1) {$S_2$};
\draw [thick, <-] (0,-2.7) -- (0,-3.8);
\draw [thick, <-] (4,-2.7) -- (4,-3.8);
\draw [thick] (-0.5,-5.3) -- (0.5,-4.8) -- (0.5,-5.8) -- (-0.5,-5.3);
\node at (0.2,-5.3) {$\mathbf{S}$};
\draw [thick] (3.5,-4.5) -- (4.5,-4) -- (4.5,-5) -- (3.5,-4.5);
\node at (4.2,-4.5) {$\mathbf{S_1}$};
\draw [thick] (3.5,-6.1) -- (4.5,-5.6) -- (4.5,-6.6) -- (3.5,-6.1);
\node at (4.2,-6.1) {$\mathbf{S_2}$};
\draw [thick] (3.5,-4.5) -- (2.5,-5.3) -- (3.5,-6.1);
\draw [thick, ->] (0.7,-5.3) -- (2.3,-5.3);
\node at (-0.5,-3.3) {$\mathfrak{F}(i)$};
\node at (4.5,-3.3) {$\mathfrak{F}(k)$};
\node at (1.5,-5.6) {$\mathscr{F}$};
\end{tikzpicture}
\end{gathered}
\end{equation}

while sequences of the form of Diagram \eqref{flow-2} can be mapped to cobordisms of the form:

\begin{equation} \label{CCCD-to-Cob-2}
\begin{gathered}
\begin{tikzpicture}[every tqft/.append style={transform shape}]
\draw[rotate=90] (0,0) ellipse (2.5cm and 1 cm);
\node[above] at (0,1.7) {$\mathscr{B}$};
\node at (-0.5,1.1) {$P$};
\node at (-0.5,-1.1) {$R$};
\begin{scope}[tqft/every boundary component/.style={draw,fill=green,fill opacity=1}]
\begin{scope}[tqft/cobordism/.style={draw}]
\begin{scope}[rotate=90]
\pic[tqft/reverse pair of pants, at={(-1,0)}, name=a];
\pic[tqft/pair of pants, anchor=incoming boundary 1, name=b, at=(a-outgoing boundary 1)];
\end{scope}
\end{scope}
\end{scope}
\draw[rotate=90] (0,-4) ellipse (2.5cm and 1 cm);
\node[above] at (4,1.7) {$\mathscr{B}$};
\node at (2,0.7) {$\mathscr{S}$};
\node at (4.4,1.1) {$Q$};
\node at (4.4,-1.1) {$R$};
\draw [thick, <-] (0,-2.7) -- (0,-3.8);
\draw [thick, <-] (4,-2.7) -- (4,-3.8);
\draw [thick] (-0.5,-4.5) -- (0.5,-4) -- (0.5,-5) -- (-0.5,-4.5);
\node at (0.2,-4.5) {$\mathbf{P}$};
\draw [thick] (-0.5,-6.1) -- (0.5,-5.6) -- (0.5,-6.6) -- (-0.5,-6.1);
\node at (0.2,-6.1) {$\mathbf{R}$};
\draw [thick] (-0.5,-4.5) -- (-1.6,-5.3) -- (-0.5,-6.1);
\node at (-1.8,-5.3) {$\mathbf{S}$};
\draw [thick] (3.5,-4.5) -- (4.5,-4) -- (4.5,-5) -- (3.5,-4.5);
\node at (4.2,-4.5) {$\mathbf{Q}$};
\draw [thick] (3.5,-6.1) -- (4.5,-5.6) -- (4.5,-6.6) -- (3.5,-6.1);
\node at (4.2,-6.1) {$\mathbf{R}$};
\draw [thick] (3.5,-4.5) -- (2.5,-5.3) -- (3.5,-6.1);
\draw [thick, ->] (0.7,-5.3) -- (2.3,-5.3);
\node at (-0.5,-3.3) {$\mathfrak{F}(i)$};
\node at (4.5,-3.3) {$\mathfrak{F}(k)$};
\node at (1.5,-5.6) {$\mathscr{F}$};
\end{tikzpicture}
\end{gathered}
\end{equation}

In either case, $\mathfrak{F}: \mathbf{CCCD} \rightarrow \mathbf{Cob}$ is the functor from the category $\mathbf{CCCD}$ of CCCDs (and hence of QRFs) to the category of $\mathbf{Cob}$ finite cobordisms required to define a TQFT.  In general, we can state:

\begin{theorem}[\cite{fgm:22a} Thm. 1] \label{thm1}
For any morphism $\mathscr{F}$ of CCCDs in $\mathbf{CCCD}$, there is a cobordism $\mathscr{S}$ such that a diagram of the form of Diagram \eqref{CCCD-to-Cob} or \eqref{CCCD-to-Cob-2} commutes.
\end{theorem}
\noindent
referring to \cite{fgm:22a} for the proof.

Theorem \ref{thm1} applies to any sequential measurement; therefore, it applies to measurements of a sector $X$ followed by measurements of the associated memory sector $Y$, or vice versa.  Assuming for convenience $dim(X) = dim(Y)$, we can consider a composite operation $Q = (\overrightarrow{Q}, \overleftarrow{Q})$, where $\overrightarrow{Q} = Q_X Q_Y$ and $\overleftarrow{Q} = Q_Y Q_X$, is then a pair of QRF sequences that can be identified with TQFTs that measure and record an outcome, mapping $\mathcal{H}_X \rightarrow \mathcal{H}_Y$, and dually use an outcome read from memory to prepare a state, mapping $\mathcal{H}_Y \rightarrow \mathcal{H}_X$, respectively as in Diagram \ref{bipartite-ops}:

\begin{equation} \label{bipartite-ops}
\begin{gathered}
\begin{tikzpicture}[every tqft/.append style={transform shape}]
\draw[rotate=90] (0,0) ellipse (2.8cm and 1 cm);
\node[above] at (0,1.7) {$\mathscr{B}$};
\draw [thick] (-0.2,1.6) arc [radius=1.6, start angle=90, end angle= 270];
\draw [thick] (-0.2,1) arc [radius=1, start angle=90, end angle= 270];
\draw[rotate=90,fill=green,fill opacity=1] (1.3,0.2) ellipse (0.3 cm and 0.2 cm);
\draw[rotate=90,fill=green,fill opacity=1] (-1.3,0.2) ellipse (0.3 cm and 0.2 cm);
\node[above] at (-2.2,-0.3) {$Q$};
\draw [ultra thick, white] (-0.8,1.3) -- (-1,1.3);
\draw [ultra thick, white] (-0.8,1.1) -- (-1,1.1);
\draw [ultra thick, white] (-0.8,0.9) -- (-1,0.9);
\draw [ultra thick, white] (-0.8,-0.9) -- (-1,-0.9);
\draw [ultra thick, white] (-0.8,-1.1) -- (-1,-1.1);
\draw [ultra thick, white] (-0.8,-1.3) -- (-1,-1.3);
\end{tikzpicture}
\end{gathered}
\end{equation}
\noindent
This composite operator $Q$ is, by Theorem \ref{thm1}, a TQFT \cite{fgm:23}.  Hence the operation of recording observational outcomes for a sector $X$ made at $t$ to memory, and then comparing them to later observations at $t + \Delta t$, is formally equivalent to propagating the ``system'' $X$ forward in time from $t$ to $t + \Delta t$.  

Identifying QRFs as ``internal'' TQFTs allows a general analysis of information exchange between multiple QRFs deployed by a single system, e.g., $A$.  Because all QRFs act on $\mathscr{B}$, information exchange between QRFs requires a channel that traverses $B$.  Any such channel is itself a QRF, one deployed by $B$.  Considering $A$ to comprise two observers, one deploying $Q_1$ and the other deploying $Q_2$, that interact via a local operations, classical communication (LOCC \cite{chitambar:14}) protocol provides an example:

\begin{equation} \label{locc-diag}
\begin{gathered}
\begin{tikzpicture}[every tqft/.append style={transform shape}]
\draw[rotate=90] (0,0) ellipse (2.8cm and 1 cm);
\node[above] at (0,1.9) {$\mathscr{B}$};
\draw [thick] (-0.2,1.9) arc [radius=1, start angle=90, end angle= 270];
\draw [thick] (-0.2,1.3) arc [radius=0.4, start angle=90, end angle= 270];
\draw[rotate=90,fill=green,fill opacity=1] (1.6,0.2) ellipse (0.3 cm and 0.2 cm);
\draw[rotate=90,fill=green,fill opacity=1] (0.2,0.2) ellipse (0.3 cm and 0.2 cm);
\draw [thick] (-0.2,-0.3) arc [radius=1, start angle=90, end angle= 270];
\draw [thick] (-0.2,-0.9) arc [radius=0.4, start angle=90, end angle= 270];
\draw[rotate=90,fill=green,fill opacity=1] (-0.6,0.2) ellipse (0.3 cm and 0.2 cm);
\draw[rotate=90,fill=green,fill opacity=1] (-2.0,0.2) ellipse (0.3 cm and 0.2 cm);
\draw [rotate=180, thick, dashed] (-0.2,0.9) arc [radius=0.7, start angle=90, end angle= 270];
\draw [rotate=180, thick, dashed] (-0.2,0.3) arc [radius=0.1, start angle=90, end angle= 270];
\draw [thick] (-0.2,0.5) -- (0,0.5);
\draw [thick] (-0.2,-0.1) -- (0,-0.1);
\draw [thick] (-0.2,-0.9) -- (0,-0.9);
\draw [thick] (-0.2,-0.3) -- (0,-0.3);
\draw [thick, dashed] (0,0.5) -- (0.2,0.5);
\draw [thick, dashed] (0,-0.1) -- (0.2,-0.1);
\draw [thick, dashed] (0,-0.9) -- (0.2,-0.9);
\draw [thick, dashed] (0,-0.3) -- (0.2,-0.3);
\node[above] at (-3.5,1.7) {$A$};
\node[above] at (2.8,1.7) {$B$};
\draw [ultra thick, white] (-0.9,1.5) -- (-0.7,1.5);
\draw [ultra thick, white] (-1,1.3) -- (-0.8,1.3);
\draw [ultra thick, white] (-1,1.1) -- (-0.8,1.1);
\draw [ultra thick, white] (-1,0.9) -- (-0.8,0.9);
\draw [ultra thick, white] (-1.1,0.7) -- (-0.8,0.7);
\draw [ultra thick, white] (-1.1,0.5) -- (-0.8,0.5);
\draw [ultra thick, white] (-1,-0.9) -- (-0.8,-0.9);
\draw [ultra thick, white] (-1,-1.1) -- (-0.8,-1.1);
\draw [ultra thick, white] (-1,-1.3) -- (-0.8,-1.3);
\draw [ultra thick, white] (-0.9,-1.5) -- (-0.7,-1.5);
\draw [ultra thick, white] (-0.9,-1.7) -- (-0.7,-1.7);
\draw [ultra thick, white] (-0.8,-1.9) -- (-0.6,-1.9);
\draw [ultra thick, white] (-0.8,-2.1) -- (-0.6,-2.1);
\node[above] at (-1.3,1.4) {$Q_1$};
\node[above] at (-1.3,-2.4) {$Q_2$};
\draw [rotate=180, thick] (-0.2,2.3) arc [radius=2.1, start angle=90, end angle= 270];
\draw [rotate=180, thick] (-0.2,1.7) arc [radius=1.5, start angle=90, end angle= 270];
\draw [thick] (-0.2,1.9) -- (0,1.9);
\draw [thick] (-0.2,1.3) -- (0,1.3);
\draw [thick, dashed] (0.2,1.9) -- (0,1.9);
\draw [thick, dashed] (0.2,1.3) -- (0,1.3);
\draw [thick] (-0.2,-1.7) -- (0,-1.7);
\draw [thick] (-0.2,-2.3) -- (0,-2.3);
\draw [thick, dashed] (0.2,-1.7) -- (0,-1.7);
\draw [thick, dashed] (0.2,-2.3) -- (0,-2.3);
\draw [ultra thick, white] (0.3,2) -- (0.3,1.2);
\draw [ultra thick, white] (0.5,2) -- (0.5,1.2);
\draw [ultra thick, white] (0.7,1.9) -- (0.7,1.1);
\draw [ultra thick, white] (0.3,-2.4) -- (0.3,-1.5);
\draw [ultra thick, white] (0.5,-2.4) -- (0.5,-1.5);
\draw [ultra thick, white] (0.7,-1.8) -- (0.7,-1.5);
\node[above] at (4.5,-2.4) {Classical channel};
\draw [thick, ->] (2.9,-2) -- (0.7,-0.8);
\node[above] at (4.5,-1.4) {Quantum channel};
\draw [thick, ->] (2.9,-0.9) -- (2.3,-0.6);
\end{tikzpicture}
\end{gathered}
\end{equation}
\noindent
In a LOCC protocol, one channel is considered ``classical'' while the other is considered ``quantum''; however, this language masks the fact that both channels are physical.  As pointed out in \cite{tipler:14}, all media supporting classical communication are physical, and interactions with these media are always local measurements or preparations.  Hence the two channels in a LOCC protocol are physically equivalent -- both are TQFTs implemented by $B$ -- although their conventional semantics are different.

Diagram \eqref{locc-diag} can, clearly also represent externally-mediated communication between any two functional components of a system, e.g., macromolecular pathways within a cell or functional networks within a brain.  We show in \cite{fgm:23} that whenever $Q_1$ and $Q_2$ are deployed by distinct -- technically, separable or mutually decoherent -- ``observers'' or ``systems,'' they fail to commute, i.e., the commutator $[Q_1, Q_2] = Q_1 Q_2 - Q_2 Q_1 \geq h/2$, where again $h$ is Planck's constant.  As shown in \cite{fg:21}, Theorem 3.4 using the CCCD representation, non-commutativity of QRFs induces quantum contextuality, i.e., dependence of measurement results on ``non-local hidden variables'' that characterize the measurement context \cite{bell:66, kochen:67, mermin:93}.  In the current context, such hidden variables characterize the action of $H_B$ on $\mathscr{B}$, affecting what $A$ will observe next in every cycle of $A$-$B$ interaction.  

As shown in \cite{Dzha2018}, such context dependence can, in principle, be captured classically if sufficient measurements of the context can be implemented.  Such measurements would, however, have to access all of $B$.  The existence of an MB prevents such access; in the current setting, $A$ has access to $B$ only via $\mathscr{B}$.   The finite energetic cost of measurement, and consequent requirement for a thermodynamic sector $F$, prevents measurement even of all of $\mathscr{B}$ by any finite physical system.  Hence, we can expect physical systems, including all biological systems, to employ only local context-dependent control to switch between mutually non-commuting (sets of) QRFs.  How context switches implemented by QRF switches induce evolution, development and learning was introduced in \cite{fgl:21}. Some specific of context switching will be discussed \S \ref{applications}.

\section{Tensor network representation of control flow} \label{tn}


\subsection{Tensor networks and holographic duality}

Entanglement and quantum error correction, two concepts developed in quantum information theory, have been proved to have a fundamental role in unveiling quantum gravity \cite{Qi:18}. At the origin of this consideration there has been the discovery by Bekenstein and Hawking \cite{Bek:72,Bek:73,Haw:71,Haw:74} that the second law of thermodynamics can be preserved in the gravitational field of a black hole, if this latter has an entropy proportional to the area of its horizon, by the inverse of the Newton gravitational constant G. This entropy is maximal, as implied by the second law itself, providing an upper bound for possible configurations of matter within a region of the same size \cite{Suss:95,Bous:02}.

Nonetheless, the scaling of the local degrees of freedom counted by the entropy does not increase as the volume, hinging toward the formulation of the holographic conjecture \cite{'tHo:93}, suggesting a division between the information that can only be retrieved on the boundary world, and a merely apparent bulk world. AdS/CFT realized the holographic conjecture, postulating a duality between gravity in asymptotically AdS space and quantum field theory on the spatial infinity of the AdS space \cite{Mal:99}. Giving  literal meaning to the duality, Ryu and Takayanagi (RT) proposed that entanglement of a boundary region fulfils the same law as for the black hole entropy, replacing the area of the black hole horizon with an extremal surface area that bounds the bulk region under scrutiny. 

While on the boundaries the theory can be individuated assigning a specific conformal field theory (CFT), in the bulk the geometry can be associated to specific entanglement structures of the quantum systems. This is for instance what happens to the ground states of a CFT associated to an AdS space: the RT area surface increases less fast than the volume of the boundary. When the boundary is at equilibrium, in a thermal state of finite temperature, the bulk geometry corresponds to that of a black hole, its horizon being parallel to the boundary and its size increasing with the temperature. The RT surface is then confined between the boundary and the back hole horizon, approaching the boundary at higher temperature and increasing its entropy. These considerations suggest the existence of a subtle link interconnecting the structure of space-time and quantum entanglement, and hence that a theory of quantum gravity must be fundamentally holographic, where its states satisfy the RT formula for some bulk geometry.  

The existence of an exact correspondence between bulk gravity and quantum theory at the boundary may hinge toward possible inconsistencies with locality. This has been discussed in the literature, in terms of local reconstruction theory \cite{Alm:15,Hu:12,Hea:14}: variables in the bulk (e.g. bulk spins) can be controlled instantaneously from the boundary, but requiring simultaneous access to a large portion of the boundary: locality and upper speed of light do not hold exactly in this theory. Nonetheless, local observers confined in small regions at the boundary still fulfil locality and the existence of an upper limit of the speed of information exchange, in a way that is reminiscent of quantum error correction code (QECC) in quantum information theory: information is stored redundantly, in such a way that when part of it is corrupted, a reconstruction of information is still possible. Locality in the bulk is therefore a QECC property of the encoding map that realizes the duality between bulk and boundary. On the other hand, these properties are strictly connected to RT, which provides the necessary resource of entanglement for QEEC to emerge.

The RT formula and QECC are properties fulfilled by different classes of models, among which TNs \cite{Swin:12}. These have been first introduced in condensed matter physics as variational wave-functions of strongly correlated systems \cite{Whi:92,Ver:08}.  TNs are many-body wavefunctions that can be derived composing few-body quantum states, which are indeed tensors. A prototype of TN is e.g. the Einstein-Podolsky-Rosen (EPR) entangled pairs of qubits: in an entangled basis, measured qbits are in some entangled pure state and can be composed with remaining ones with increasing complexity: complicated quantum entanglement can be derived by only entangling a few qubits \cite{Di:99}.

Particularly relevant for its implications on the reconstruction of the space-time structure is the multi-scale entanglement renormalization ansatz (MERA) \cite{Vid:08}. TNs can be naturally related to holography duality by considering that their entanglement entropy can be controlled by their graph geometry. Some versions of TNs that are characterized by RT entanglement entropy and QEEC have been constructed resorting to stabilizer codes \cite{Pa:15,Yang:16} and random tensors with large bond dimension \cite{Hay:16}. TNs with random tensors at each node can be regarded as random states restricted by the topology of the network. Exactly as random states are almost maximally entangled, random TNs show through the RT formula an almost maximal entanglement, providing a large family of states with interesting properties to explore holographic duality. Furthermore, for random TNs the RT formula holds in generic spaces with not necessarily hyperbolic geometry, hinging toward an extension of holographic duality beyond AdS, to more general configurations in quantum gravity. Nonetheless, at least in three dimensions, random tensor networks have been related to the gravitational action, by means of the Regge calculus \cite{Han:17}.

On the other hand, since geometry emerges as a specification of the entanglement structure, one may consider that the Einstein equations should be connected as well to the dynamics of entanglement. For small perturbations around the ground state of a CFT in boundary, linearized Einstein equations have been derived from the RT formula \cite{Lash:14,Swin:14}. Indeed, the conformal symmetry enables a relations between the energy momentum and the entanglement entropy, and consequently the area of the extremal surface c an be connected to the energy-momentum distribution at the boundary --- this is equivalent to the linearized Einstein equations.

The dynamics on the boundary, on the other side, shows a chaotic behaviour, with scrambling of the single-particle operators, which evolve into multi-particle operators \cite{Shen:14}. Maximal chaotic behaviour recovered in the ladder operators commutator growth, is encoded in the out-of-time-ordered correlation (OTOC) functions, characterized by exponential growth in time and temperature. A model endowed with this properties is e.g. the Sachdev-Ye-Kitaev model, developed to describe certain systems in condensed matter physics, such as Gapless spin-fluid \cite{Sach:93,Ki:15,Mal:16}. On the other hand, operator scrambling is also related to QEEC: the chaotic dynamics at the boundary instantiates QECC preserving quantum information, efficiently hidden (and protected) behind the horizon. Nevertheless this has led to may questions, concerning the information behind the horizon being eventually accessible from the boundary though non-local measurements, the fate of the local degrees of freedom hitting the singularity, the relation among the causal structure of the bulk and a smooth geometry across the horizon.

\subsection{General results}

We can move to prove a general result:

\begin{theorem} \label{main-thm}
A system $A$ exhibits non-trivial control flow if, and only if, its control flow can be represented by a TN.
\end{theorem}
\noindent
and examine some of its corollaries.  We begin by defining:

\begin{definition}
Control flow is {\em trivial} if a system deploys only one QRF.
\end{definition}
\noindent
As any collection of mutually-commuting QRFs can be represented as a single QRF \cite{fg:21, fgm:22a}, any system that deploys only mutually-commuting QRFs exhibits trivial control flow.

Systems that deploy only a single QRF ``do the same thing'' regardless of context, and so do not qualify as ``interesting'' in the sense used here.  As noted above, no finite physical system can measure the entire state of its boundary with a single QRF, so no such system can simultaneously measure and act on its entire context.  Any system $A$ that deploys multiple QRFs in sequence cannot, as noted above, avoid contextuality due to unobservable effects, mediated by the action of $H_B$, of the action of $Q_i$ on the state measured by $Q_j$.  Every action taken by an ``interesting'' system, in other words, at least transiently increases the VFE at its boundary.

Consider, then, a system $A$ that deploys multiple, distinct QRFs $Q_1, Q_2, \dots Q_n$, where $n \ll N = dim(H_{AB})$.  Classical control flow in $A$ can then be represented by a matrix $\mathbf{CF} = [P_{ij}]$, where $P_{ij}$ is the probability of the control transition $Q_i \rightarrow Q_j$.  As noted earlier, any such transition has an energetic cost, which must be paid with free energy sources from $F$.

The matrix $\mathbf{CF}$ is a 2-tensor.  Theorem \ref{main-thm} states that this tensor can be decomposed into a TN.  We prove it as follows:

\begin{proof}[Proof (Thm. \ref{main-thm})]
Suppose first that control flow in a system $A$ can be represented by a TN.  A TN is, by definition, a factorization of a tensor operator into a network of tensor operators.  This network can be either hierarchical or flat; if it is hierarchical, each layer can be considered a flat TN.  Hence no generality is lost in considering just the case of a flat TN, which is an operator contraction $T = \dots T_{ij} T_{jk} T_{kl} \dots$, where summation on shared indices is left implicit.   In general, $T_{jk} \neq T^T_{jk} = T_{kj}$, hence these expressions do not commute.  They therefore represent non-trivial control flow.  Conversely, any non-trivial control flow can be written, at any fixed scale or level of abstraction, as a linear sequence of (in general probabilistic) operators.  The fixed order of operators in the sequence can be encoded formally by adding ``spatial'' indices as needed to allow contraction over shared indices.  Hence any non-trivial control flow at a fixed scale can be written as a flat TN.  This construction can be repeated at each larger scale to produce a hierarchical TN over a collection of ``lowest-scale'' TNs.
\end{proof}

We can now examine two corollaries of this result:

\begin{corollary} \label{decoherence-corrol}
Decoherent reference sectors exist on $\mathscr{B}$ if and only if control flow can be implemented by a TN.
\end{corollary}

\begin{proof}
Decoherence sectors require independently-deployable, non-commuting QRFs.  This requires a control structure that factors, by Theorem \ref{main-thm} a TN.  Conversely, a TN factors the control structure, making QRFs independently deployable, which renders their sectors decoherent.
\end{proof}

Equivalently, the GM factors if, and only if, control flow can be implemented by a TN.

\begin{corollary} \label{cycle-corrol}
The TN of any system compliant with the FEP is a decomposition of the Identity.
\end{corollary}

\begin{proof}
The FEP applies to systems with a NESS, and drives such systems to return to (the vicinity of) the NESS after any perturbation.  Hence at a sufficiently large scale, the TN of any such system is a cycle, i.e., a decomposition of the Identity.
\end{proof}

Many standard TN models, e.g., MERAs, assume boundary conditions asymptotically far, in numbers of lowest-scale operators, from the region of the network that is of interest.  Identifying such asymptotic boundary conditions yields a cyclic system.

Theorem \ref{main-thm}, together with its corollaries, provides a natural, formal means of classifying systems by their control architectures.  At a high level, two characteristics distinguish systems with different architectures:

\begin{itemize}
\item Hierarchical depth indicates the number of ``virtual machine'' layers \cite{smith:05} the architecture supports.  The interfaces between these layers implement coarse-graining, removing from the higher-level representation all dimensions, and hence all information, which is contracted out of the lower-level operators.
\item Number and location of contractions that yield unitary operators, and hence build in entanglement between lower-level operators.   The natural limit is a MERA, in which every pair of lower-level operators is entangled at every hierarchical level \cite{orus:19}.  
\end{itemize} 
\noindent
The control-flow architecture, in turn, specifies the structure of the ``layout'' of distinguishable sectors on $\mathscr{B}$ and hence of detectable features/objects in the environment.  Locality on $\mathscr{B}$ requires a hierarchical TN; detectable entanglement requires a MERA-like TN.  Locality is required for detectable features/objects to appear to have components with nested decomposition.  Any QRF for geometric space, and hence for spacetime, must be hierarchical, and must be a MERA if entanglement in space is to be detected.  A MERA is required, in particular, if the use of coherence between spatially-separated systems as a computational or communication resource is detectable.

To illustrate the classification of systems by hierarchical level, consider the ten-step cyclic TN shown in Diagram \eqref{TN-1}:

\begin{equation} \label{TN-1}
\begin{gathered}
\begin{tikzpicture}
\draw [thick] (0,0) circle [radius=0.3];
\draw [thick] (1,0) circle [radius=0.3];
\draw [thick] (3,0) circle [radius=0.3];
\draw [ultra thick] (0.3,0) -- (0.7,0);
\draw [ultra thick] (1.3,0) -- (1.7,0);
\node at (2, 0) {$\dots$};
\draw [ultra thick] (2.3,0) -- (2.7,0);
\node at (0,0) {$A$};
\node at (1,0) {$B$};
\node at (3,0) {$J$};
\draw [ultra thick] (-0.3,0) to [out=170,in=180] (0,0.8) to [out=0,in=180] (1.5,0.8) to [out=0,in=180] (3,0.8) to [out=0,in=10] (3.3,0);
\end{tikzpicture}
\end{gathered}
\end{equation}
\noindent
and its extension to a hierarchy as shown in Diagram \eqref{TN-2}:

\begin{equation} \label{TN-2}
\begin{gathered}
\begin{tikzpicture}
\draw [thick] (0,0) circle [radius=0.3];
\node at (0,0) {$A$};
\draw [ultra thick] (0.3,0) -- (0.7,0);
\draw [thick] (1,0) circle [radius=0.3];
\node at (1,0) {$B$};
\draw [ultra thick] (1.3,0) -- (1.7,0);
\draw [thick] (2,0) circle [radius=0.3];
\node at (2,0) {$C$};
\draw [ultra thick] (2.3,0) -- (2.7,0);
\draw [thick] (3,0) circle [radius=0.3];
\node at (3,0) {$D$};
\draw [ultra thick] (3.3,0) -- (3.7,0);
\draw [thick] (4,0) circle [radius=0.3];
\node at (4,0) {$E$};
\draw [ultra thick] (4.3,0) -- (4.7,0);
\draw [thick] (5,0) circle [radius=0.3];
\node at (5,0) {$F$};
\draw [ultra thick] (5.3,0) -- (5.7,0);
\draw [thick] (6,0) circle [radius=0.3];
\node at (6,0) {$G$};
\draw [ultra thick] (6.3,0) -- (6.7,0);
\draw [thick] (7,0) circle [radius=0.3];
\node at (7,0) {$H$};
\draw [ultra thick] (7.3,0) -- (7.7,0);
\draw [thick] (8,0) circle [radius=0.3];
\node at (8,0) {$I$};
\draw [ultra thick] (8.3,0) -- (8.7,0);
\draw [thick] (9,0) circle [radius=0.3];
\node at (9,0) {$J$};
\draw [thick, red] (-1,1) circle [radius=0.3];
\draw [thick, red] (1,1) circle [radius=0.3];
\draw [thick, blue] (2,1) circle [radius=0.3];
\draw [thick, red] (3,1) circle [radius=0.3];
\draw [thick, blue] (4,1) circle [radius=0.3];
\draw [thick, red] (5,1) circle [radius=0.3];
\draw [thick, blue] (6,1) circle [radius=0.3];
\draw [thick, red] (7,1) circle [radius=0.3];
\draw [thick, blue] (8,1) circle [radius=0.3];
\draw [thick, blue] (10,1) circle [radius=0.3];
\draw [ultra thick, red] (0.2,0.2) -- (0.8,0.8);
\draw [ultra thick, blue] (1.2,0.2) -- (1.8,0.8);
\draw [ultra thick, red] (2.2,0.2) -- (2.8,0.8);
\draw [ultra thick, blue] (3.2,0.2) -- (3.8,0.8);
\draw [ultra thick, red] (4.2,0.2) -- (4.8,0.8);
\draw [ultra thick, blue] (5.2,0.2) -- (5.8,0.8);
\draw [ultra thick, red] (6.2,0.2) -- (6.8,0.8);
\draw [ultra thick, blue] (7.2,0.2) -- (7.8,0.8);
\draw [ultra thick, red] (1.2,0.8) -- (1.8,0.2);
\draw [ultra thick, blue] (2.2,0.8) -- (2.8,0.2);
\draw [ultra thick, red] (3.2,0.8) -- (3.8,0.2);
\draw [ultra thick, blue] (4.2,0.8) -- (4.8,0.2);
\draw [ultra thick, red] (5.2,0.8) -- (5.8,0.2);
\draw [ultra thick, blue] (6.2,0.8) -- (6.8,0.2);
\draw [ultra thick, red] (7.2,0.8) -- (7.8,0.2);
\draw [ultra thick, blue] (8.2,0.8) -- (8.8,0.2);
\draw [ultra thick] (-0.3,0) to [out=170,in=180] (0,-1.2) to [out=0,in=180] (4,-1.2) to [out=0,in=180] (9,-1.2) to [out=0,in=10] (9.3,0);
\draw [ultra thick, red] (-0.2,0.2) -- (-0.8,0.8);
\draw [ultra thick, blue] (9.2,0.2) -- (9.8,0.8);
\draw [ultra thick, red] (-1,0.7) to [out=-90,in=180] (-0.6,-0.6);
\draw [ultra thick, red] (-0.4,-0.6) to [out=0,in=180] (7.5,-0.6) to [out=0,in=-90] (8,-0.3);
\draw [ultra thick, blue] (1,-0.3) -- (1,-0.5);
\draw [ultra thick, blue] (1,-0.7) to [out=-90,in=180] (1.3,-0.9) to [out=0,in=180] (9.3,-0.9);
\draw [ultra thick, blue] (9.5,-0.9) to [out=0,in=-90] (10,0.7);
\draw [ultra thick, green] (2,2) circle [radius=0.3];
\draw [ultra thick, green] (1.2,1.2) -- (1.8,1.8);
\draw [ultra thick, green] (2.2,1.8) -- (2.9,0.25);
\draw [ultra thick, green] (5,2) circle [radius=0.3];
\draw [ultra thick, green] (4.2,1.2) -- (4.8,1.8);
\draw [ultra thick, green] (5.2,1.8) -- (5.8,1.2);
\draw [ultra thick, green] (8,2) circle [radius=0.3];
\draw [ultra thick, green] (8,1.3) -- (8,1.7);
\draw [ultra thick, green] (8,2.3) to [out=90,in=0] (7.5,2.7) to [out=180,in=0] (0.8,2.7) to [out=180,in=90] (0,0.3);
\end{tikzpicture}
\end{gathered}
\end{equation}
\noindent
where red, blue, and green colors indicate distinct hierarchical ``layers'' of tensor contractions.  We have trained artificial neural networks (ANNs) to execute these TNs as the sequences of state transitions shown in Table 1.  The first sequence (Dataset 1) is a ten-step cycle shown Diagram \eqref{TN-1}; the second sequence (Dataset 2) layers the coarse-grained state transitions of Diagram \eqref{TN-2} onto this ten-step cycle.  In Dataset 2, a two-bit tag is used to differentiate the ``low-level'' from the coarse-grained ``high-level'' cycles.  An example state state transition from a randomly-generated initial state is shown in Fig. \ref{transition}; the red-on-green bit pattern effectively moves ``up'' one step on each state-transition cycle.

\begin{figure}[H]
\centering
\includegraphics[width=13 cm]{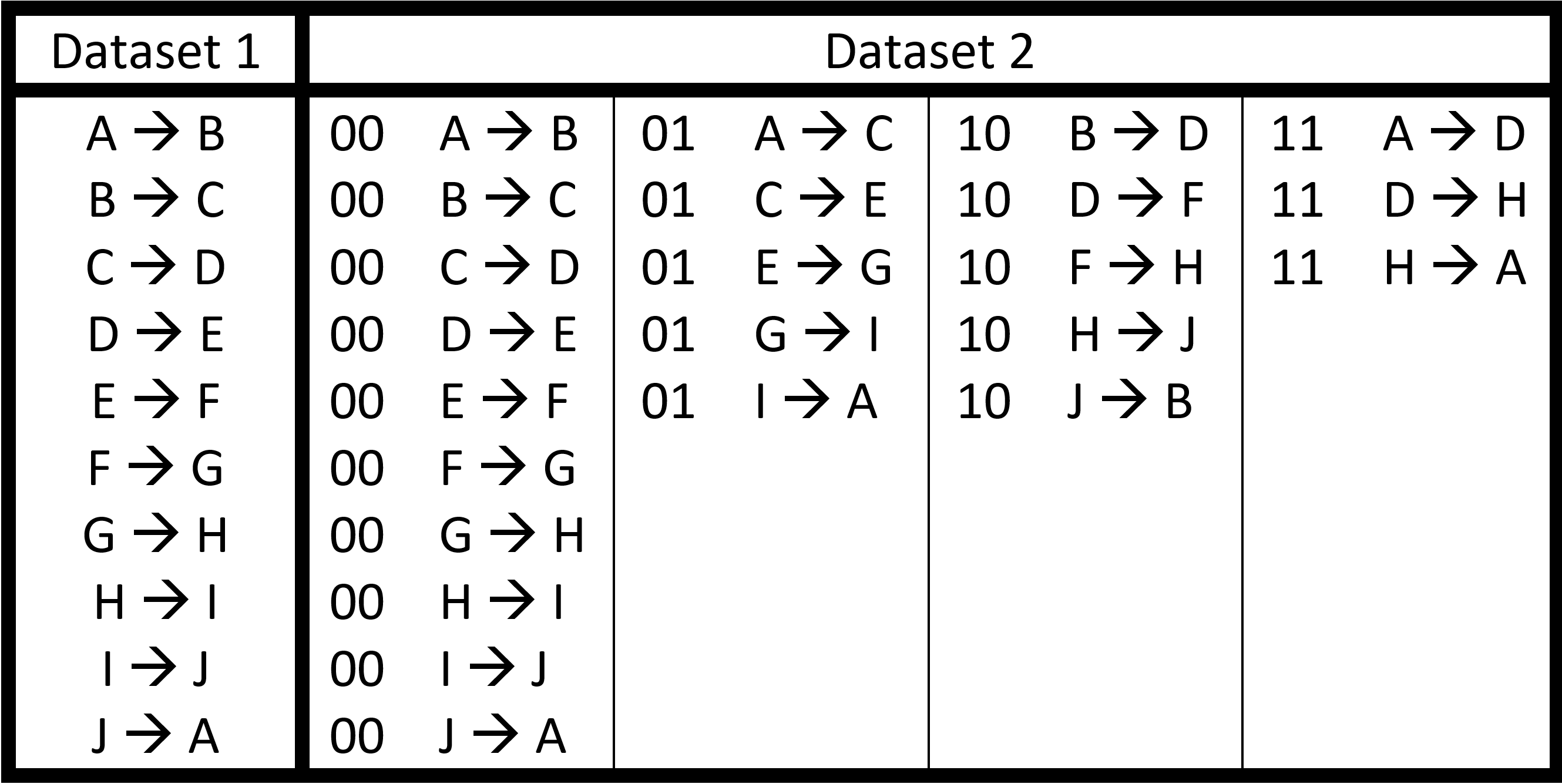}
\label{table-1}
\end{figure}
\begin{quote}
Table 1: Datasets used in ANN simulations.  Dataset 1 specifies a ten-step cycle A $\rightarrow$ B $\rightarrow ~\dots ~\rightarrow J \rightarrow A$.  Dataset 2 specifies this same cycle, with three coarse-grained cycles layered on top.  The tags (0,0), (0,1), (1,0), and (1,1) distinguish the data for the low- and high-level cycles.
\end{quote}
~\\
\begin{figure}[H]
\centering
\includegraphics[width=11 cm]{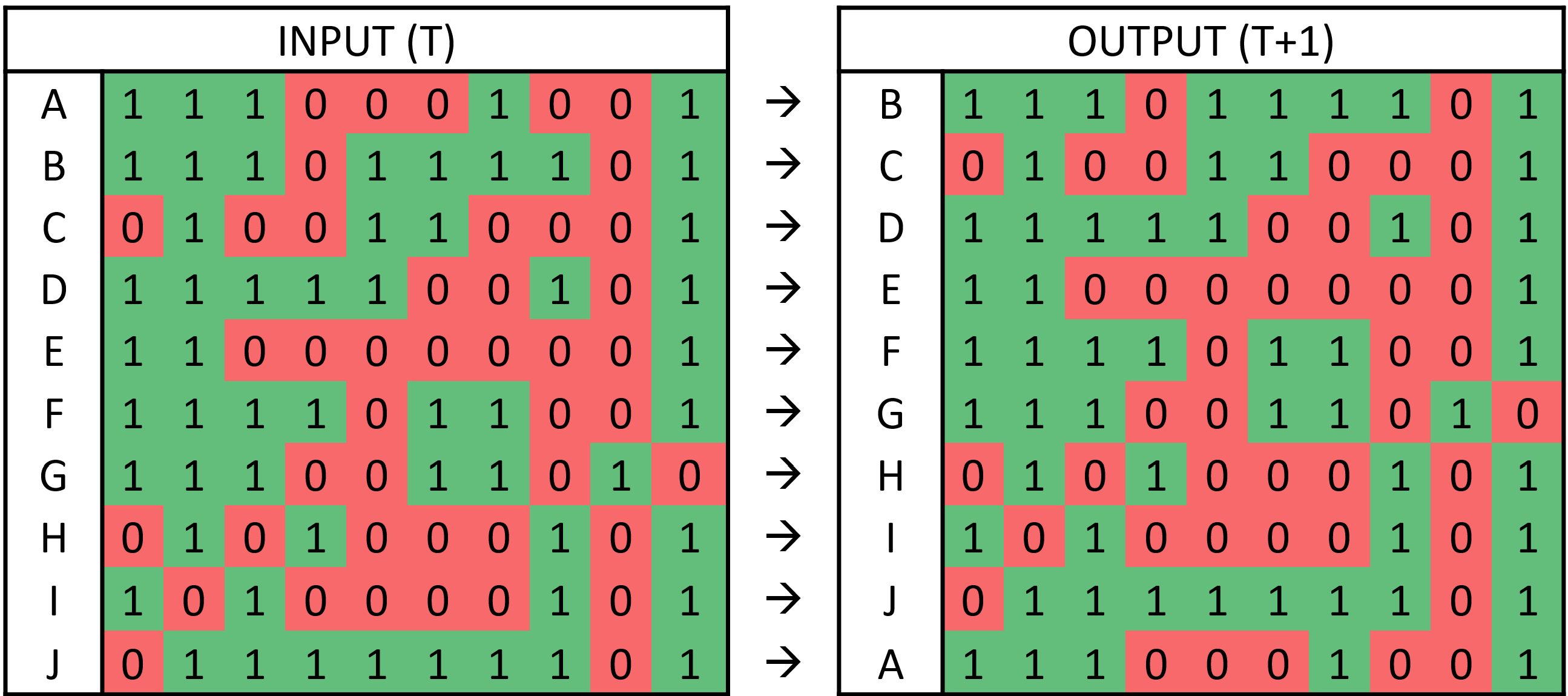}
\caption{Example state transition from Dataset 1.}
\label{transition}
\end{figure}

We trained two ANNs, one to execute each of the control cycles shown in Table 1.  The networks are each composed of three layers, as illustrated in Fig. \ref{ff-arch}, with network sizes of [10, 50, 10] and [10, 200, 10], respectively, for the input, hidden, and output layers.  The units in the hidden layer use the rectified linear unit (ReLU) nonlinear activation function and the neurons in the output layer use the hyperbolic tangent activation function. The network is connected in a feedforward way where a neuron in one layer connects to every neuron in the next layer. Since the ANN serves as a switch state controller, we use a training scheme, similar to one-class classification \cite{manevitz:02}, where the training data are the only data that the network learns to produce. In so doing, the network learns to overfit the training data, and any input outside of the designated state-encoding is discarded.  The network is, therefore, not expected to deviate from the learned pattern.  The network learns both control regimes wth 100\% accuracy after training with 3,000 randomly-generated 10-bit inputs. 

\begin{figure}[H]
\centering
\includegraphics[width=11 cm]{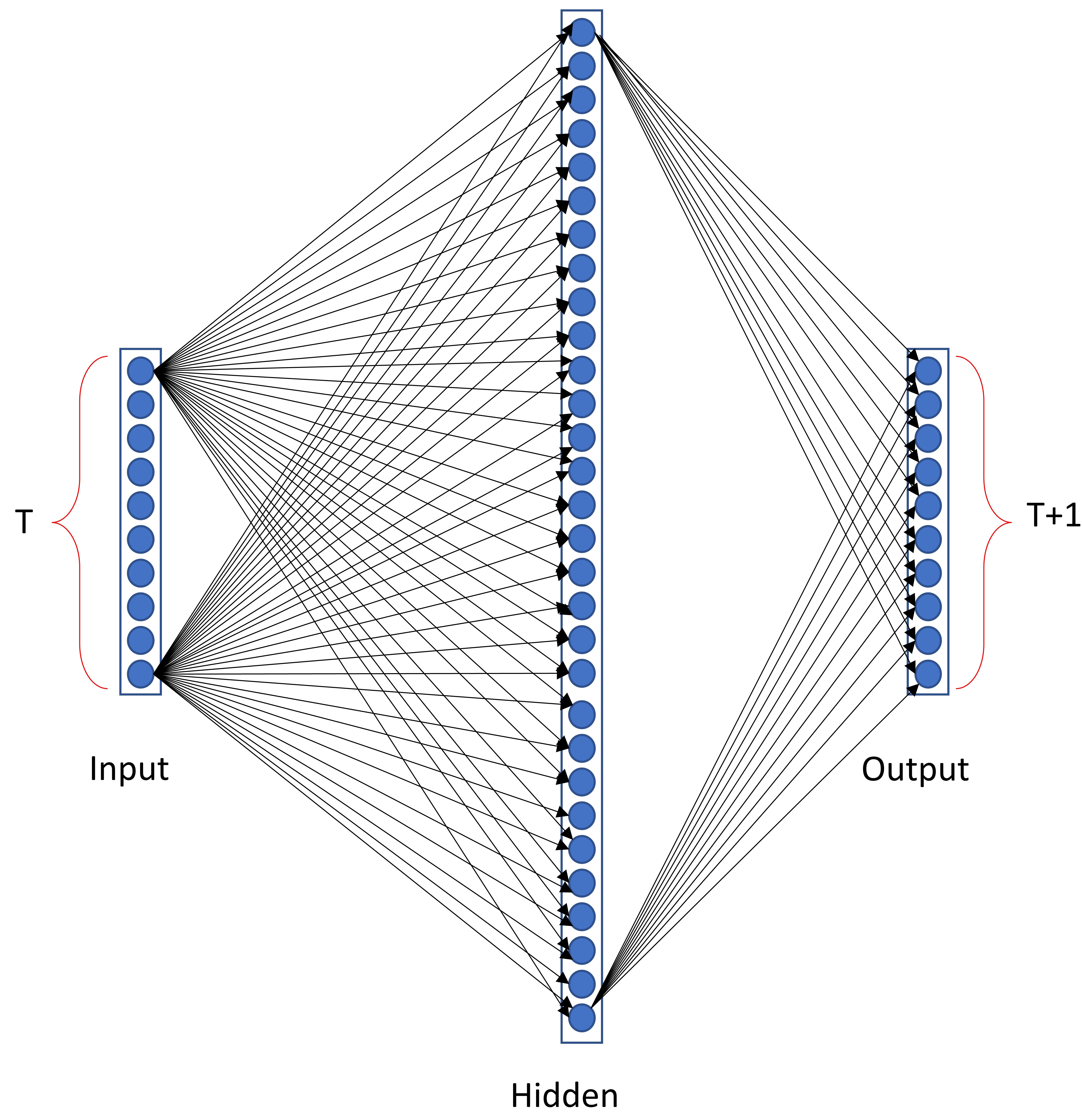}
\caption{Feed-forward network archtecture used to learn the control cycles specified in Table 1.  Each node is connected to every node of the next layer, as shown here for the first and last nodes only.  The labels `T' and `T+1' indicate time steps in the executed control flow. }
\label{ff-arch}
\end{figure}

In the more realistic case of noisy input data, where binary states can be flipped, the Bidirectional Associative Memory (BAM), a minimal two-layer nonlinear feedback network \cite{kosko:88}, is a viable alternative to a shallow feed-forward ANN. The architecure is shown in Fig. \ref{bam-arch}.  This BAM network learns to associate between the two initial and final states in Table 1, with similar performance to that of the feed-forward network.

\begin{figure}[H]
\centering
\includegraphics[width=11 cm]{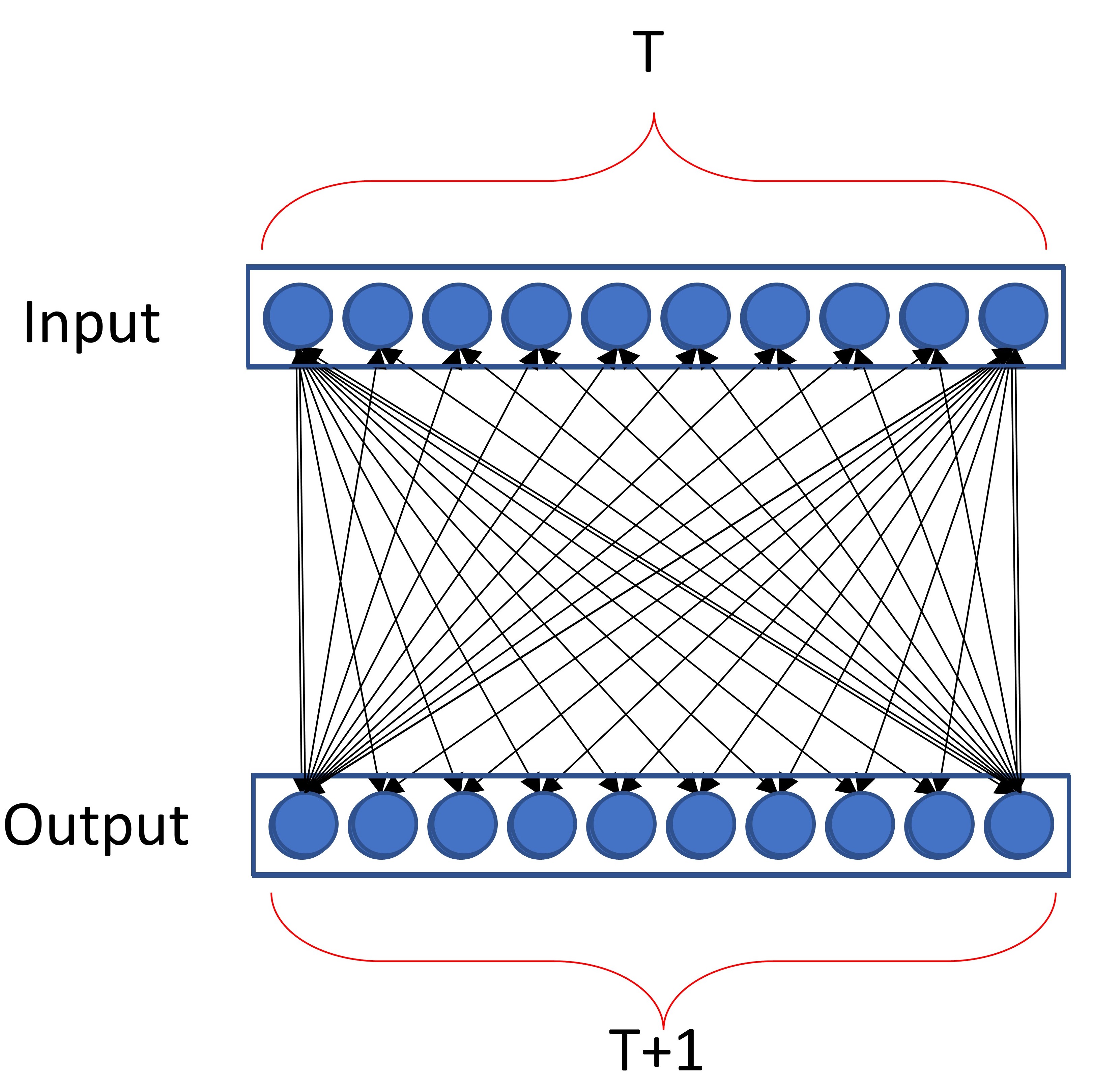}
\caption{Architecture of the Bidirectional Associative Memory (BAM) network employed here. As in Fig. \ref{ff-arch}, only the connections of the first and last nodes are shown explicitly.}
\label{bam-arch}
\end{figure}

\section{Implementing control flow with TQNNs} \label{tqnn}

%
%
%

Tensor Networks can be naturally associated to the matrix elements of physical scalar products among topological quantum neural networks (TQNNs). Physical scalar products encode indeed the dynamics of TQFTs, since they fulfill their constraints of imposing flatness of the curvature and gauge invariance. Thus, the matrix elements associated to scalar products can be seen as evolution matrix elements for the spin-network states that span the Hilbert spaces of TQNNs.

\subsection{Tensor networks as classifiers for TQNNs } \label{TQNNs}

A notable example is provided by BF theories \cite{Ba:96}, a class of TQFTs particularly well studied in the literature of mathematical physics that enables expressing effective theories of particle physics, gravity and condensed matter, and provides as well a general framework for implementations of models of quantum information and quantum computation, machine learning (ML) and neuroscience. These are defined on the principal bundle $M$ of a connection $A$ for some internal gauge group $G$, with algebra $\mathfrak{g}$, according to the action on a \emph{d}-dimensional manifold $\mathcal{M}_d$
\begin{equation}
\mathcal{S}=\int_{\mathcal{M}_d} {\rm Tr} [B \wedge F]\,,
\end{equation}
where $B$ is an ad$(\mathfrak{g})$-valued  \emph{d-2}-form, $F$ denotes the field-strength of $A$, which is a \emph{2}-form, and the trace Tr is over the internal indices of $\mathfrak{g}$, ensuring gauge invariance of the density Lagrangian $\mathcal{L}={\rm Tr} [B \wedge F]$ of the BF theory.

Variation with respect to the conjugated variables, the connection $A$ and the $B$ frame-field, closing a canonical symplectic structure, provide the equations of motion of the theory \cite{Ba:96}:
\begin{equation}
F=0\,, \qquad d_AB=0\,,
\end{equation}   
respectively the curvature constraint, imposing the flatness of the connection, and the Gau\ss\, constraint, imposing invariance under gauge transformations, having denoted with $d_A$ the covariant derivative with respect to the connection $A$.

At the quantum level, the states of the kinematical Hilbert space of the theory, fulfilling by construction the Gau\ss\, constraint, can be represented in terms of cylindrical functionals $Cyl$, supported on graphs $\Gamma$ that are unions of segments $\gamma_i$, the end points of which meet in nodes $n$, and with holonomies -- elements of the group $G$ -- $H_{\gamma_i}[A]$ of the connection $A$ assigned to $\gamma_i$ and intertwiner operators -- invariant tensor product of representations -- $v_n$ assigned to the nodes $n$.

For $G =$ SU(2), spin-networks $|\Gamma, j_\gamma, \iota_n \rangle$, supported on $\Gamma$ and labelled by the spin $j_\gamma$ of the irreducible representations of the group elements assigned to $\gamma$ and by the quantum intertwiner numbers $\iota_n$ associated to $v_n$, represent a basis of the kinematical Hilbert space of the theory. In terms of functionals of $Cyl$, one can provide the holonomy representation, which is related to the ``spin and intertwiner'' representation of $|\Gamma, j_\gamma, \iota_n \rangle$ by means of the Peter-Weyl transform. This allows us to decompose the spin-network cylindric functional as \cite{Ro:04}:
\begin{equation}
\Psi_{j_{\gamma_{ij}},\iota_{n_i} }(h_{\gamma_{ij}})=\left(\bigotimes_n \iota_n \right)\cdot \left(\bigotimes_{\gamma_{ij}} D^{(j_{\gamma_{ij}})} (h_{\gamma_{ij}})\right)
\,,
\end{equation}
with $D^{(j)}$ are Wigner matrices providing representation matrices of the SU$(2)$ group elements.

The functorial evolution among spin-networks is ensured by the projector operator \cite{fgm:22}, which implements the curvature constraint in the physical scalar product among states, i.e.
\begin{equation}
\langle {\rm in} |  P \,  |{\rm out} \rangle\,, \qquad {\rm with} \qquad P= \int \mathcal{D} [N] \exp(\imath\int {\rm Tr}[ N F] ) \,.
\end{equation}
We may then regard $| {\rm in} \rangle $ as elements of the Hilbert space, and without loss of generality pick up those ones resulting from composing tensorially in $Cyl$ \emph{k}-representations of holonomies. We may further denote them as $|j_1 \dots j_k \rangle $, with some ordering prescription to associate the topological structure of $\Gamma$ to the sequence of spin labels. Physically evolving states $P | {\rm in} \rangle $ are distinguished from the former ones by labelling them as  $|\widetilde{j_1 \cdots j_k} \rangle $. Similarly, we introduce $| {\rm out} \rangle $ as the tensor product of  (\emph{n-k})-representations of holonomies, and denote these states as  $|i_1 \dots i_{n-k} \rangle $. Then the matrix elements of $\langle {\rm in} |  P \, | {\rm out} \rangle$ naturally give rise \cite{fgm:23} to an \emph{n}-tensor, i.e.
\begin{equation}
\langle i_1 \dots i_{n-k} |\widetilde{j_1 \cdots j_k} \rangle =T_{i_1\dots i_{n-k} j_1 \dots j_k}\,.
\end{equation}

\subsection{Geometric RG flow for TQNNs and TNs}

The mathematical structures of TQNNs we summarized in Sec.~\ref{TQNNs} are picturing systems ``at equilibrium'', for which TQFTs characterize a topological stability that percolates into the related transition amplitudes. Nonetheless, it is worth considering as well how stochastic noise might interfere with the topological order ensured by TQFTs, and study the role of ``out-of-equilibrium'' physics in the analysis of the evolution of the systems under scrutiny.

Out-of-equilibrium dynamics is instantiated considering a heat-flow evolution of the fundamental fields of the theory, with respect to a thermal time $\tau$.  Typical Langevin equations, complemented with stochastic noise, provide through their convergence toward the equations of motion of the theory the relaxation toward equilibrium of the field configurations representing specific systems \cite{Pa:81}. In general, given some fields $\phi_\sigma$, with a classical equation of motion derived, according to the variational principle ${\delta \mathcal{S}}/{\delta \phi_\sigma}$, from an action $\mathcal{S}$ over a Euclidean manifold $\mathcal{M}$, the associated Langevin equations read:
\begin{equation}
\frac{\partial}{\partial \tau} \phi_\sigma = -\frac{\delta \mathcal{S}}{\delta \phi_\sigma} 
+\eta_\sigma   \,, 
\end{equation}
with $\eta_\sigma $ a stochastic noise term. The theory at equilibrium is characterized by the symmetries of the equations of motion ${\delta \mathcal{S}}/{\delta \phi_\sigma}=0$ that are broken in the transient phase \cite{Lulli:2021bme}; these symmetries are consistent with -- and in the case of BF theories, actually generated by -- the theories at equilibrium.

A prototype of geometric heat-flow was introduced by Hamilton, and then used by Perelman to prove the Poincar\'e conjecture, which goes under the name of Ricci flow. Here the gravitational field $g_{\mu \nu}$ is the basic configurational space field, while the drift terms are the Einstein equations of motion in the vacuum, which indeed are expressed by requiring that the components of the Ricci tensor vanish, i.e. $R_{\mu \nu}=0$. 
The Ricci flow then reads 
\begin{equation}
\imath \frac{\partial}{\partial \tau} g_{\mu \nu}= -2 R_{\mu \nu}\,,
\end{equation}   
having considered now a Lorentzian manifold $\mathcal{M}$. The Ricci flow equations can be further complemented introducing the Ricci target $R^{T}_{\mu \nu}=\kappa^2 (T_{\mu\nu}-1/2 g_{\mu \nu }T) $, expressed in terms of the Newton constant $G=\kappa^2/(8 \pi )$ and the energy-momentum tensor of matter $T_{\mu\nu}$, so as to obtain at equilibrium the Einstein equations: 
\begin{equation}
R_{\mu\nu}-\frac{1}{2} g_{\mu\nu}R= \kappa^2 T_{\mu \nu}\,, \qquad {\rm or \, \, equivalently} \qquad R_{\mu \nu}=R^{T}_{\mu \nu} \,.
\end{equation}      
The stochastic version of the Ricci flow, with heat equation turning into a Langevin equation, has been introduced and deepened in \cite{Lulli:2021bme} for a generic gravitational system in the presence of matter fields, describing an action $\mathcal{S}$ for gravity and matter. Moving then from:
\begin{equation}
\imath \frac{\partial}{\partial \tau} g_{\mu \nu}= -\frac{1}{\kappa^2} \frac{\delta \mathcal{S}}{ \delta g^{\mu \nu}} +\eta\, g_{\mu \nu} \,,
\end{equation}   
in which a multiplicative noise $\eta_{\mu \nu}=\eta\, g_{\mu \nu}$ appears, the Hamiltonian analysis of the stochastic Ricci flow (SRF) in the Adomian decomposition method (ADM) variables has been derived \cite{Lulli:2021bme}.

An essential by-product of the discussion, from the Ricci flow perspective, is that the equilibration trajectories corresponds to those of a renormalization group (RG) flow. The thermal time $\tau$ plays the role of scale parameter that individuates a dimension in the bulk, which is out-of-equilibrium. The boundaries are recovered asymptotically in $\tau$, in the infra-red regime, and are by definition at equilibrium and thus symmetric.

For a particular class of TQFTs, the BF theories we have introduced in Sec.~\ref{TQNNs} for implementing TQNNs and TNs, the geometric RG flow acquire a specific expression as the TQFT equivalent of the gravitational Ricci flow \cite{Marciano202X}.

\subsection{TNs as a generalization of the main model architectures in ML}

The use of TNs is an emerging topic in the ML community. The integration between the two appears quite immediate. A TN structure can be viewed as an ML model in which the parameters are properly adjusted to learn the classification of a data set. Yet, as Ref.~\cite{Sen:22} mentions, machine learning can aid, in turn, in determining a factorization of a TN approximating a data set. Moreover, TNs are also used to compress the layers of ANN architectures, besides a variety of other uses. Tensor networks are becoming more and more popular to the extent that they are a powerful tool for representing and manipulating high-dimensional data, as in the case of image and video classification tasks in which the data is represented as a high-dimensional tensor. High efficiency, flexibility, and easy to use are making them a dominant choice for many AI applications. Furthermore, besides being used to represent data, TNs can be used to process data by exploiting a number of operators. This feature makes them an effective technique for processing data in ML applications.

As it is well known, TNs are particularly well suited for representing quantum many-body states in which the dimension of the Hilbert space is exponentially large in the number of particles. The corresponding ML approach consists in:

\begin{itemize}
\item
Lifting data to exponentially higher spaces;
\item
Applying any linear classifier $f(x)=W^\star \Phi(X)$ to a non-linear space;
\item
Compressing the weights by using TNs.
\end{itemize}

The output of the model is a separation of classes that would not be linearly separable in a linear space. In particular, the decision function is the overlap of the weight tensor $W$ with the feature map tensor $\Phi$ as in Fig.~\ref{F1}.  The weight tensor $W$ can be approximated by the decomposition in Fig.~\ref{F2}.

\begin{figure}[H]
\centering
\includegraphics[width=11 cm]{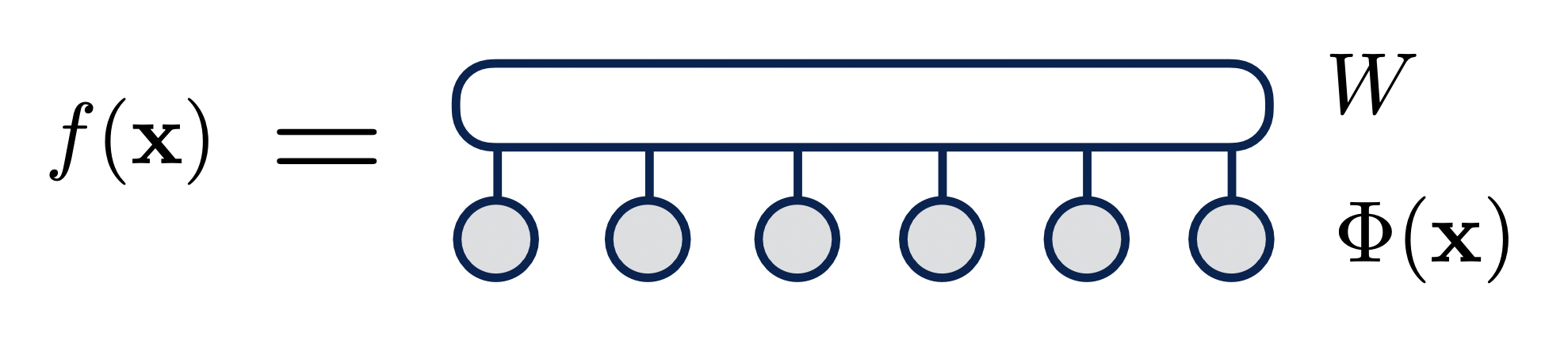}
\caption{Representation of the decision function (see \cite{Stoudenmire:18}).}
\label{F1}
\end{figure}

\begin{figure}[H]
\centering
\includegraphics[width=11 cm]{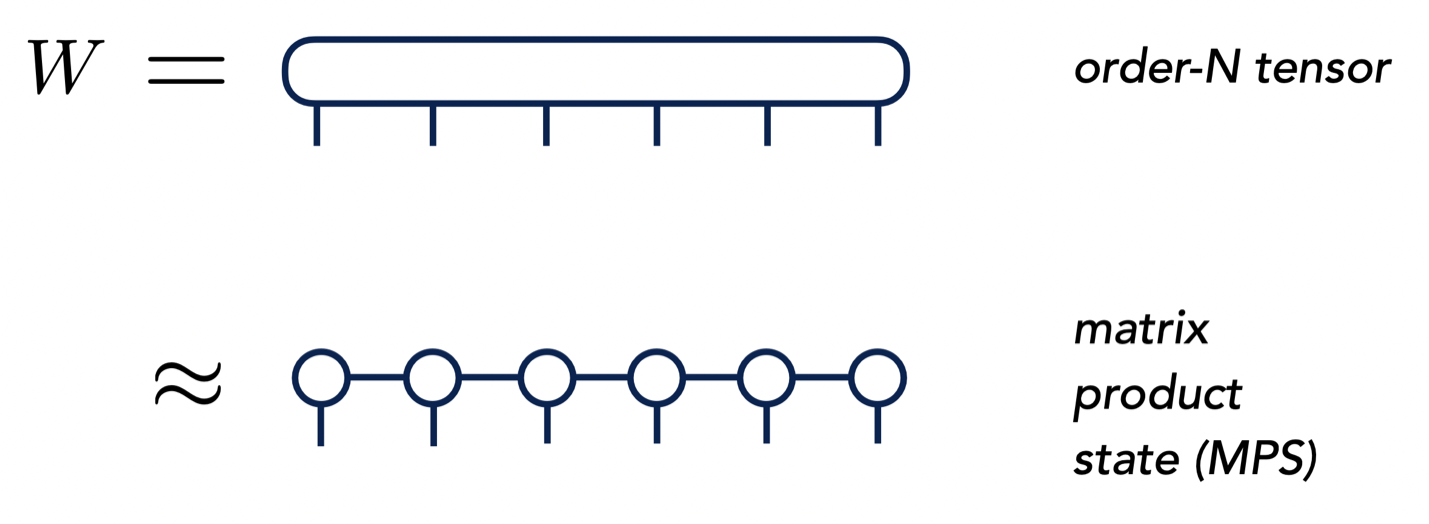}
\caption{Matrix product decomposition (again see \cite{Stoudenmire:18}).}
\label{F2}
\end{figure}

Regularization and optimization are built as a constructive product of low-order tensors while weight compression is performed by using the Matrix Product States (MPS) decomposition. If we look at Deep Neural Networks as a piecewise composition of linear discriminators (logistic regression functions), then the TN framework appears as a generalization of the main model architectures found in the ML literature, e.g. Support Vector Machines, Kernel models, and Deep Neural Networks.

The literature concerning the use of tensor theory in traditional ML is becoming large. A short review starts with a seminal paper by Stoudenmire and Schwab \cite{Stou:16}, which demonstrated how algorithms for optimizing TNs can be adapted to supervised learning tasks by using MPS (tensor trains) to parametrize non-linear kernel learning models. Novikov, Trofimov, and Oseledets \cite{No:17} have shown how an exponentially large tensor of parameters can be represented in a factorized format called Tensor Train (TT), with the consequence of obtaining a regularization of the model. van Glasser, Pancotti, and Cirac \cite{Glas:18} explored the connection between TNs and probabilistic graphical models by introducing the concept of a ``generalized tensor network architecture'' for ML. Ref.~\cite{Han:18} then designed a generative model, i.e. a traditional machine learning model that learns joint probability distributions from data and generates samples according to it, by using MPS. Ref.~\cite{Liu:21} made use of autoregressive MPSs for building an unsupervised learning model that goes beyond proof-of-concept by showing performance comparable to standard traditional models. Finally, Ref.~\cite{Con:22} analyzes the contribution of polynomials of different degrees to the supervised learning performance of different architectures.

\section{Implications for biological control systems} \label{applications}

Scale-free biology requires a smooth transition from quantum-like to classical-like behavior.  Typical representations of metabolic, signal-transduction, and gene-regulatory pathways are entirely classical, even though many of their steps involve electron-transfer or other mechanisms that are acknowledged to require a quantum-theoretic description \cite{zwier:10, groenhof:13}.  As noted earlier, free-energy budget considerations suggest that both prokaryotic and eukaryotic cells employ quantum coherence as a computational resource \cite{fl:21}.  Emerging empirical evidence for longer-range entanglement in mammalian brains suggests that large-scale networks may also be using quantum coherence as a resource \cite{kerskens:22}.  Control flow models must, therefore, support the possibility of quantum computation in biological systems.  Hierarchical TNs that include unitary components, e.g.,  MERA-type models, provide this capability.  

In prokaryotes, the primary tasks of control flow are adapting metabolism to available resources via metabolite-driven gene regulation \cite{ledesma:21} and initiating DNA replication and cell division when conditions are favorable.  We can, therefore, expect shallow hierarchies of effectively classical control transitions in these organisms.  Eukaryotes, however, are characterized by both intracellular compartmentalization and morphological degrees of freedom at the whole-cell scale.  We have shown previously that the FEP will induce ``neuromorphic'' morphologies -- i.e. morphologies that segregate inputs from outputs and enable a fan-in/fan-out computational architecture -- in any systems with morphological degrees of freedom \cite{ffglm:22}.  Such systems can be expected to have deep control hierarchies at the cellular level, with hierarchical structure correlating with morphological structure in morphologically-complex cells such as neurons \cite{fgl:22}, and in multicellular assemblages at all scales.  As well as managing metabolism and replication, such systems must implement active exploration of the environment, communication with other systems, and -- crucially for cognition -- the writing and reading of stigmergic memories.  Thus we can expect such systems to implement QRFs for spacetime and for specific kinds of objects, e.g., conspecifics and suitable substrates for recording stigmergic memories.  Such QRFs rely on symmetries, and hence on redundancy of encoded (or encodable) information; they depend, in other words, on the availability of error-correcting codes \cite{fl:17, fl:22}.  The implementation of spacetime as a quantum error-correcting code by TNs has been extensively studied by physicists; see \cite{bain:20} for review and \cite{fgm:23} for a detailed analysis using the present formalism.  The use of spacetime as an error-correcting code by organisms -- e.g., the implementation of translational and rotational invariance of objects by dorsal visual processing in mammals \cite{flombaum:08, fields:11} -- is well-understood phenomenologically, but the details of neural implementation remain to be elucidated. 

Both the context-sensitivity of, and the occurence of context effects due to non-commutativity of QRFs in, control networks can be expected to increase with their complexity and hierarchical depth.  ``Bowtie'' networks with high fan-in/fan-out to/from multi-use proteins or second messengers such as Ca$^{2+}$ are increasingly recognized as ubiquitous in high eukaryotic cells \cite{niss:20}.  Such networks have the general for of the CCCD depicted in Diagram \eqref{cccd-2}.  Frequently, such networks evolve via compression of information (e.g. toward share second messengers, as in [Ca$^{2+}$]-based interactions \cite{carafoli:16,pol:09}) as an efficiency-increasing mechanism. Bowties introduce semantic ambiguities that must be resolved by context.  Each incoming signal has its own governing semantics, but the relevant context can depend on boundary conditions which can be exceedingly difficult (if not impossible) to predetermine (see e.g., \cite{friedlander:15, boniolo:21} for general discussions of the history and semantic depth of this problem).  As pointed out in \cite{fgl:21}, a context change $x \mapsto y$ is semantically problematic if for a fixed set $\{ o_i \} $ of observations, the conditional probability distributions $\rm{P}(o_i\vert x)$ and $\rm{P}(o_i \vert y)$ are well defined, but the joint distribution $\rm{P}(o_i \vert x \vee y)$ is not \cite{kochen:67}.  This occurs whenever the QRFs for $x$ and $y$ do not commute \cite[Th 7.1]{fg:21}.  As suggested by Diagram \eqref{cccd-2}, this context-switching problem affects deep learning using VAEs \cite{kingma:19}; see e.g., the application to antimicrobial peptides in \cite{dean:20}. In general, the structure of Diagram \eqref{cccd-2} can serve as a convenient benchmark for distinguishing signal transduction networks that incorporate co-deployable versus non-co-deployable QRFs \cite{fg:21}.

``Quantum'' context effects due to non-commutativity have, interestingly, been reported even at the scale of human language use.  The ``Snow Queen'' experiment \cite{cervantes:18} challenged subjects with distinct, mutually-inconsistent meanings of terms such as `kind', `evil', or `beautiful' in different contexts, and detected statistically-significant context effects using the CbD formalism \cite{Dzha2017a, Dzha2018}.  Such effects cannot be explained by linguistic ambiguity, misreading, etc.  Such language-driven contextuality is taken up in the setting of psycholinguistics and distributional semantics in \cite{wang:21}, which combines CbD and the sheaf theoretic \cite{abramsky:11, abramsky:17} methods to systematically study semantic ambiguity as creating meaning/sense discrepancies in statements like ``It was about time'', ``She had time on her hands to win the heat'', ``West led with a queen'', etc.

While the notion of ``languages'' has thus far been applied to cells, tissues, and even non-vertebrates in a mostly metaphorical way, we can speculate that linguistic approaches to understanding the interplay between context dependence and semantic ambiguity may be useful to biology in general.  Immune cells (e.g., T cells) are, for example, ``programmed'' or ``trained'' by their progenitor cells to respond to local cellular signals and ambient conditions in particular ways.  Unexpected context changes may induce dysfunctional (at the organism scale) responses, including chronic disorders \cite{focus:19}; these can be considered consequenes of discrepancies between the ``actual' semantics of incoming signals and the semantics expected by the immune systems' ``language.''  This suggestion of possible ``linguistic'' contextuality seems in consonance with the hypothesis of \cite{atlan:98} that the immune system is a cognitive (living) system implementing its exclusive system of language-grammar, which may be prone to analogous disorders of communication as those discussed in \cite{wang:21}. Similar context effects have been observed in microbiological systems \cite{basieva:11}; here discrepancies in experimentally derived classical probabilities arising from lactose-glucose interference signaling in {\em E. Coli} can only be explained in terms of non-classical probabilities. We note that the expression `quantum-like' \cite{khrennikov:15} is often used for such effects; however, their formal structure is exactly that given by quantum theory.

We expect that further research into quantum biology will unfold significant perspectives on human/mammalian physiology and cognitive capabilities along the lines suggested in the present article. For example, allostatic maintenance, as briefly alluded to in \S\ref{attractor}, can be seen as a process regulating a body's physiological conditions relative to costs and benefits while dynamically allocating resources for the purpose of overall adaptability of an organism within its internal environment.  Implementing the allostatic and anticipatory mechanisms are the visceromotor cortical regions generating autonomic, hormonal, and immunological predictions leading to interoceptive inference \cite{barrett:16,corcoran:20,hohwy:16,seth:16,barrett:15,barrett:17,hohwy:13,kleckner:17}.  This process of inference in humans and mammals putatively utilizes predictive coding for the processes of homeostasis-allostasis through a hierarchy of cellular to organ-level systems, in turn  connecting interoception to the processes of extercoception and proprioception \cite{seth:16,barrett:15,kleckner:17,seth:12,seth:13}. The basic principles follow from how allostasis provides protection against potential surprise by utilizing a framework somewhat beyond the error signaling necessary for homeostatic maintenance (it is essentially through minimizing the free energy of internal state trajectories towards combatting surprise, as discussed in \S\ref{attractor}). The net effect of the process is consonant with the Good Regulator theorem of \cite{conant:70}, showing how regulation of a given system requires an internal model of that system. A further perspective is to emphasize the predictive nature of an integrated, complex, allostatic-interoceptive cortical system capable of supporting a spectrum of psychological phenomena including memory and emotions \cite{kleckner:17} (cf. \cite{seth:16}). Accordingly, cognitive conditions such as depression and autism have been described as abnormalities of allostatic-interoceptive inference, so impairing predictive coding mechanisms due to aberrant assimilation and mistuning of prediction errors (putatively a connectivity issue), conceivably leading to a root cause of many known cognitive conditions \cite{seth:16,barrett:15,seth:12}.

We anticipate that this fully general, context sensitive model of control flow will be important for understanding morphogenesis, which is not simply a feed-forward emergent system, but rather a highly context-sensitive error-minimizing process \cite{levin:22}. Specifically, the collective intelligence of cells during embryonic development, organ regeneration, and metamorphosis can create and repair specific complex structures despite a wide range of perturbations \cite{pezzulo:16}. Changes in the genome, the number of cells, or the starting configuration can often be overcome: bisected embryos result in normal twins, amputated salamander limbs re-grow back to normal, and planarian fragments result in perfect little worms \cite{sanchez:08}. The competency of cellular collectives to reach the correct target morphology despite even drastic interventions requires an understanding of how they navigate, via context-sensitive control flow, problem spaces including anatomical morphospace \cite{levin:22a}, physiological, and transcriptional spaces \cite{fl:22, biswas:23}. Understanding the navigation policies used by unconventional collective intelligences can help not only understand creative problem-solving on rapid timescales (such as the ability to regulate genes to accommodate an entirely novel stressor \cite{emmons-bell:19} without evolutionary adaptation), but may also have implications for predicting and managing the goals and behavioral repertoires of synthetic beings \cite{clawson:22}.

\section{Conclusion} \label{next}

We have shown here how the problem of defining control flow arises in active inference systems, and provided three formal representations of the problem.  We have proved that control flow in such systems can always be represented by a tensor network, provided illustrative examples, and shown how the general formalism of topological quantum neural networks can be used to implement a general model of control flow.  These results provide a general formalism with which to characterize context dependence in active inference systems at any scale, from that of macromolecular pathways to that of multi-organism communities.  They suggest that the concept of communication by language is not just metaphorical when applied to biological systems in general, but rather an appropriate and productive description of interactional dynamics.

We view these results as a further step toward fully integrating the formal models, concepts, and languages of physics, biology, and cognitive science.  This integration is not reductive.  It rather allows us to classify systems using natural measures of organizational and computational complexity, and to understand how interactions between simpler systems can implement the more complex behavior of the larger systems that they compose.

\section*{Acknowledgements}
K.F. is supported by funding for the Wellcome Centre for Human Neuroimaging (Ref: 205103/Z/16/Z), a Canada-UK Artificial Intelligence Initiative (Ref: ES/T01279X/1) and the European Union’s Horizon 2020 Framework Programme for Research and Innovation under the Specific Grant Agreement No. 945539 (Human Brain Project SGA3).M.L. gratefully acknowledges funding from the Guy Foundation and the John Templeton Foundation, Grant 62230. A.M. wishes to acknowledge support by the Shanghai Municipality, through the grant No. KBH1512299, by Fudan University, through the grant No. JJH1512105, the Natural Science Foundation of China, through the grant No. 11875113, and by the Department of Physics at Fudan University, through the grant No. IDH1512092/001.

\section*{Conflict of interest}
The authors declare no competing, financial, or commercial interests in this research.

\end{document}